\documentclass[12pt]{iopart}

\usepackage{amsthm}
\usepackage{bbm}
\usepackage{amssymb,braket,appendix}
\usepackage[english]{babel}
\usepackage{graphicx}
\usepackage{color}
\usepackage{ulem}
\usepackage{tikz}
\usepackage{algpseudocode,algorithm}
\usetikzlibrary{decorations.pathreplacing, shadows, shapes}
\usetikzlibrary{shapes,arrows}

\def\ketbra#1#2{{\vert#1\rangle\!\langle#2\vert}}

\newcommand{\eu}{\mathrm{e}}
\newcommand{\im}{\mathrm{i}\,}

\newcommand{\vecx}{\mathbf{x}}
\newcommand{\statex}{\mathbf{x}}
\newcommand{\matA}{\mathbf{A}}
\newcommand{\opA}{\mathbf{A}}
\newcommand{\opMD}{\mathbf{M}_D}
\newcommand{\opD}{\mathbf{D}}
\newcommand{\matH}{\mathbf{H}}
\newcommand{\opH}{\mathbf{H}}
\newcommand{\opM}{\mathbf{M}}
\newcommand{\bigO}{\mathcal{O}}
\newcommand{\Ord}[1]{\mathcal{O}\left( #1 \right)}
\newcommand{\tOrd}[1]{\tilde{\mathcal{O}}\left( #1 \right)}

\definecolor{darkblue}{RGB}{63,63,168}
\definecolor{darkgreen}{RGB}{60,140,20}
\definecolor{darkblue}{RGB}{63,63,168}
\definecolor{mydarkgreen}{RGB}{5,89,3}
\definecolor{mydarkred}{RGB}{153, 0,0}
\definecolor{Ms}{rgb}{0.,0.3,0.7}
\definecolor{Q}{rgb}{0.9,0.,0.}
\definecolor{Pr}{rgb}{0.4,0.3,0.9}
\definecolor{LW}{rgb}{0.5,0.5,0.}

\renewcommand{\ms}[1]{{\color{Ms} #1}}

\theoremstyle{plain}
\newtheorem{oracle}{Oracle}
\newtheorem{problem}{Problem}

\newtheorem{lemma}{Lemma}
\newtheorem{result}{Result}

\newtheorem{theorem}{Theorem}
\newtheorem{assume}{Assumption}

\begin{document}

\title[Quantum gradient descent and Newton's method]{Quantum gradient descent and Newton's method for constrained polynomial  optimization}

\author{Patrick Rebentrost}
\ead{pr@patrickre.com}
\address{Research Laboratory of Electronics,
Massachusetts Institute of Technology, Cambridge, MA 02139}
\address{Xanadu, 372 Richmond St W, Toronto, M5V 2L7, Canada}
\author{Maria Schuld}
\ead{schuld@ukzn.ac.za}
\address{Quantum Research Group, School of Chemistry and Physics, University of KwaZulu-Natal, Durban 4000, South Africa}
\author{Leonard Wossnig}
\address{Institute for Theoretical Physics, ETH Zurich, 8093 Zurich, Switzerland}
\address{Department of Materials, University of Oxford, Parks Road, Oxford OX1 3PH, United Kingdom}
\address{Department of Computer Science, University College London, Gower Street, London WC1E 7JE, United Kingdom}
\author{Francesco Petruccione}
\address{Quantum Research Group, School of Chemistry and Physics, University of KwaZulu-Natal, Durban 4000, South Africa}
\address{National Institute for Theoretical Physics, KwaZulu-Natal, South Africa}
\author{Seth Lloyd}
\address{Research Laboratory of Electronics, Massachusetts Institute of Technology, Cambridge, MA 02139}
\address{Department of Mechanical Engineering, Massachusetts Institute of Technology, Cambridge, MA 02139}

\vspace{10pt}
\begin{indented}
\item[]\today
\end{indented}

\begin{abstract}
Optimization problems in disciplines such as machine learning are commonly solved with iterative methods. Gradient descent algorithms find local minima by moving along the direction of steepest descent while Newton's method takes into account curvature information and thereby often improves convergence. Here, we develop quantum versions of these iterative optimization algorithms and apply them to polynomial optimization with a unit norm constraint. In each step, multiple copies of the current candidate are used to improve the candidate using quantum phase estimation, an adapted quantum principal component analysis scheme, as well as quantum matrix multiplications and inversions. The required operations perform polylogarithmically in the dimension of the solution vector and exponentially in the number of iterations.
Therefore, the quantum algorithm can be beneficial for high-dimensional problems where a small number of iterations is sufficient.
\end{abstract}

\maketitle

\section{Introduction}
Optimization plays a vital role in various fields.
In machine learning and artificial intelligence, common techniques such as regression, support vectors machines, and neural networks rely on optimization.
Often in these cases the objective function to minimize is a least-squares loss or error function $f(\vecx)$ that takes a vector-valued input to a scalar-valued output \cite{Sra12}.
For strictly convex  functions on a convex set, there exists a unique global minimum, and, in the case of equality-constrained quadratic programming, solving the optimization problem reduces to a matrix inversion problem.
In machine learning, one often deals with either convex objective functions beyond quadratic programming, or non-convex objective functions with multiple minima. In these situations, no single-shot solution is known in general and one usually
has to resort to iterative search in the landscape defined by the objective function.

One popular approach focusses on gradient descent methods such as the famous backpropagation algorithm in the training of neural networks. Gradient descent finds a local minimum starting from an initial guess by iteratively proceeding along the negative gradient of the objective function. Because it only takes into account the first derivatives of $f(\vecx)$, gradient descent may involve many steps in cases when the problem has an unfortunate landscape \footnote{For example where it features long and narrow valleys. A famous example is the Rosenbrock function, for which gradient descent fails completely.}.  For such objective functions, second-order methods, which model the local curvature and correct the gradient step size, have been shown to perform well~\cite{martens2010deep}.
One such method, the so-called Newton's method, multiplies the inverse Hessian to the gradient of the function. By taking into account the curvature information in such a manner, the number of steps required to find the minimum often greatly reduces at the cost of computing and inverting the matrix of the second derivatives of the function with respect to all coordinates. Once the method arrives in the vicinity of a minimum, the algorithm enters a realm of {\it quadratic convergence}, where the number of correct bits in the solution doubles with every step~\cite{boyd04, Nocedal2006}.

As quantum computation becomes more realistic and ventures into the field of machine learning, it is worthwhile to consider in what way optimization algorithms can be translated into a quantum computational framework. Optimization has been considered in various implementation proposals of quantum computing \cite{Durr1996,Farhi2016}. The adiabatic quantum computing paradigm \cite{Farhi00} and its famous sibling, quantum annealing, are strategies to find the ground state of a Hamiltonian and can therefore be understood as `analogue' algorithms for optimization problems. The first commercial implementation of quantum annealing, the D-Wave machine, solves certain quadratic unconstrained optimization problems and has been tested for machine learning applications such as classification \cite{Denchev12, Neven09} and sampling for the training of Boltzmann machines \cite{Benedetti15}. In the gate model of quantum computation, quadratic optimization problems deriving from machine learning tasks such as least-squares regression \cite{Wiebe12, Schuld16} and the least-squares support vector machine \cite{Rebentrost14} were tackled with the quantum matrix inversion technique \cite{Harrow09}, demonstrating the potential for exponential speedups for such single-shot solutions under certain conditions. Variational methods that use classical optimization while computing an objective function on a quantum computer have become popular methods targeting near-term devices with limited coherence times \cite{peruzzo14, farhi14}.

In this work, we provide quantum algorithms for iterative optimization, specifically the gradient descent and Newton's methods. We thereby extend the quantum machine learning literature by techniques that can be used in non-quadratic convex or even non-convex optimization problems.
The main idea is that at each step we take multiple copies of a quantum state
$\ket{\vecx^{(t)}}$ to produce multiple copies of another quantum state
$\ket{\vecx^{(t+1)}}$ by considering the gradient vector and the Hessian matrix of the objective function. Depending on the step size taken, this quantum state improves the objective function.

We consider optimization for the special case of polynomials with a normalization constraint. The class of polynomials we discuss in detail for optimization contains homogeneous polynomials of even order over large-dimensional spaces with a relatively small number of monomials. We also sketch the polynomial optimization with a relatively small number of inhomogeneous terms added. We show how to make the gradient of the objective function and the Hessian matrix operationally available in the quantum computation using techniques such as density matrix exponentiation \cite{Lloyd14}. The objective function is assumed to be given via oracles providing access to the coefficients of the polynomial and allowing black-box Hamiltonian simulation as described in~\cite{Berry2012}.

Since at each step we consume multiple copies of the current solution to prepare a single copy of the next step, the algorithms scale exponentially in the number of steps performed. While this exponential scaling precludes the use of our algorithms for optimizations that require many iterations, it is acceptable in cases when only a few steps are needed to arrive at a reasonably good solution, especially in the case of the Newton method, or when performing a local search.
We note that the computation of gradients on a quantum computer has been investigated before, for a setting in which the coordinates of the input vectors to the function are encoded as binary numbers rather than as the amplitudes of a quantum system \cite{jordan05}.
We also note subsequent work on quantum gradient descent in \cite{Kerenidis2017} and
\cite{Gilyen2017}.

\section{Problem statement}

Gradient descent is formulated as follows. Let $f: \mathbb{R^N} \rightarrow \mathbb{R}$ be the objective function one intends to minimize. Given an initial point $\vecx^{(0)} \in \mathbb{R}^N$, one iteratively updates this point using information on the steepest descent of the objective function in a neighborhood of the current point,
\begin{equation}
\vecx^{(t+1)} = \vecx^{(t)} - \eta\; \nabla f(\vecx^{(t)}),
\label{Eq:graddesc}
\end{equation}
where $\eta > 0$ is a hyperparameter (called the learning rate in a machine learning context) which may in general be step-dependent. Newton's method extends this strategy by taking into account information about the curvature, i.e.~the second derivatives of the objective function, in every step. The iterative update therefore includes the Hessian matrix $\matH$ of the objective function,
\begin{equation}
\vecx^{(t+1)} = \vecx^{(t)} - \eta \matH^{-1} \nabla f(\vecx^{(t)}),
\label{Eq:newton}
\end{equation}
with the entries $H_{ij} = \frac{\partial^2 f}{\partial x_i \partial x_j}$ evaluated at the current point $\vecx^{(t)}$.

\begin{figure}
\centering
\includegraphics[scale = 0.5]{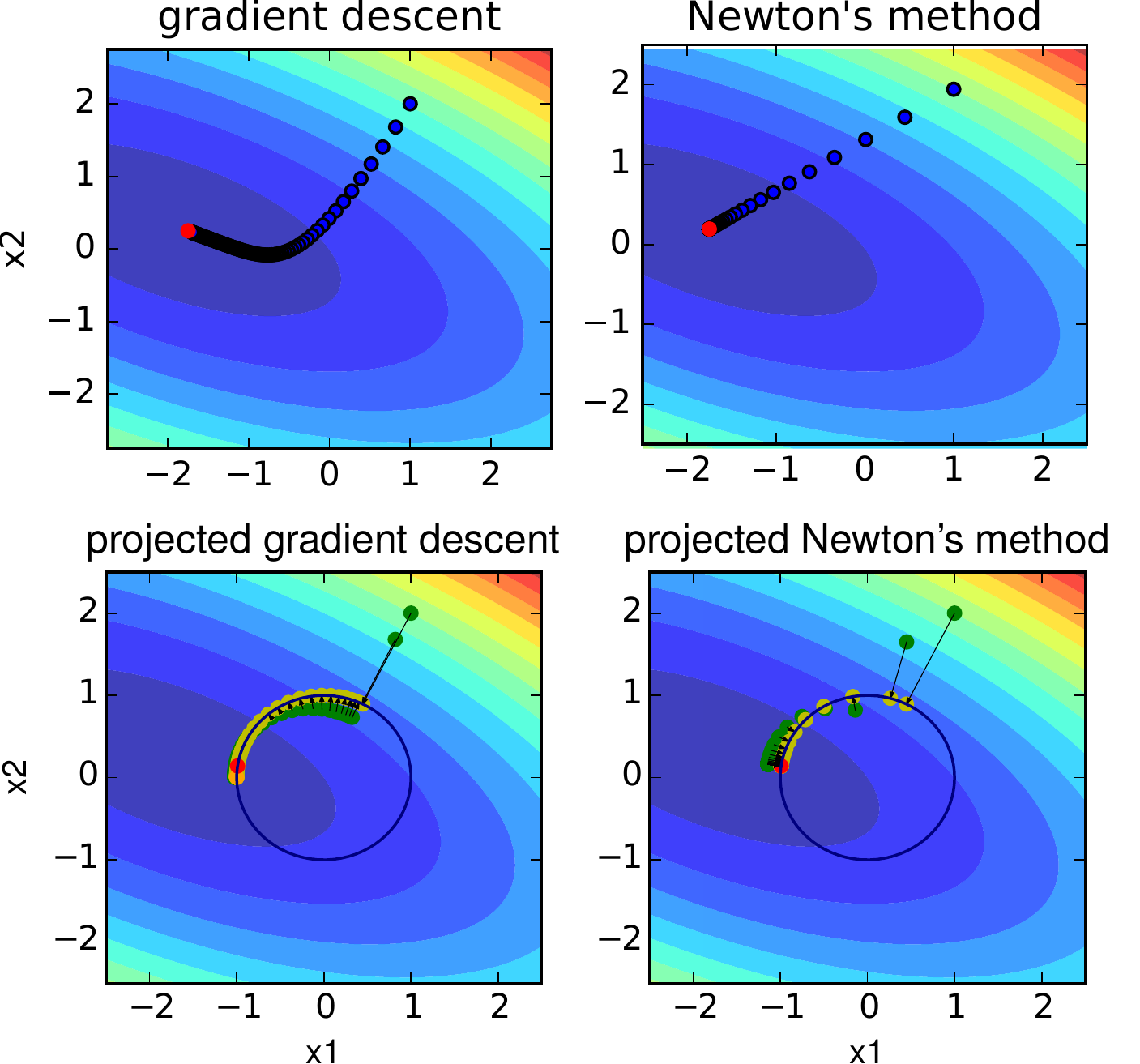}
\caption{Examples of gradient descent and Newton's method, as well as their projected versions. While gradient descent follows the direction orthogonal to the contour lines to find the minimum marked in red (upper left), Newton's method takes the curvature into account to choose a straighter path (upper right). Projected descent methods under unit sphere constraints renormalize (lower left and right, yellow dots) the solution candidate (green dots) after every update and find the minimum on the feasible set. Considered here is a quadratic objective function $f(\vecx) = \vecx^T \matA \vecx + \mathbf{c}^T \vecx$ with $a_{11} = a_{12} = a_{21} = 0.2$, $a_{22} = 0.6$ and $\mathbf{c} = (0.3,0.2)^T$, as well as step size $\eta = 0.2$ and initial state $\vecx^{(0)} = (1,2)^T$.}
\label{Fig:standard_descent}
\end{figure}

These iterative methods are in principle applicable to any sufficiently smooth function. In this work, we restrict ourselves to the optimization of multidimensional homogeneous polynomials of even order and under spherical constraints, where the polynomials only consists of a small number of terms (a property which we refer to as ``sparse"). We also discuss the extension to a small number of \textit{inhomogeneities}.

\subsection{Sparse, even homogenous polynomials}
The \textit{homogeneous} objective function we seek to minimize here is a polynomial of order $2p$ defined over $ \vecx \in \mathbb{R}^N$,
\begin{equation}
f(\vecx) = \frac{1}{2}  \sum\limits_{i_1,\dots,i_{2p}=1}^{N}    A_{i_1\dots i_{2p}}  x_{i_1}  \dots  x_{i_{2p}} ,
\label{Eq:of1}
\end{equation}
with the $N^{2p}$ coefficients $A_{i_1\dots i_{2p}} \in \mathbb{R}$ and $\vecx = (x_1, \dots, x_N)^T$. As we will see, the number of non-zero coefficients, or ``sparsity", will appear in the final resource analysis of the quantum algorithm. By mapping the inputs to a higher dimensional space constructed from $p$ tensor products of each vector, we can write this function as an \textit{algebraic form},
\begin{equation}
f(\vecx) = \frac{1}{2} \; \vecx^{T} \otimes \cdots \otimes \vecx^{T} \matA \, \vecx\otimes \cdots \otimes \vecx.
\label{Eq:of2}
\end{equation}
The coefficients become entries of a $p$-dimensional tensor $\matA  \in \mathbb{R}^{N \times N} \otimes \cdots \otimes \mathbb{R}^{N \times N}$ that can be written as a $N^p \times N^p$ dimensional matrix. Let $\Lambda_A >0$ be given such that the matrix norm of $\opA$ is $\Vert \opA \Vert \leq \Lambda_A$. Equations (\ref{Eq:of1}) and (\ref{Eq:of2}) describe a homogeneous polynomial of even order. For the more familiar case of $p=1$, the objective function reduces to $f(\vecx) = \vecx^T \matA \vecx$ and a common quadratic optimization problem. For $p=2$ and $N=2$, the two-dimensional input $\vecx = (x_1, x_2)^T$ gets projected to a vector containing the polynomial terms up to second order $\vecx \otimes \vecx = (x_1^2, x_1x_2, x_2x_1, x_2^2)^T $, and $\matA$ is of size $4\times 4$.

It will be helpful to formally decompose $\matA$ into a  sum of tensor products on each of the $p$ spaces that make up the `larger space',
\begin{equation}
\matA = \sum_{\alpha=1}^K \matA^{\alpha}_1 \otimes \cdots \otimes \matA^{\alpha}_p,
\label{Eq:dectensA}
\end{equation}
where each $\matA^{\alpha}_i$ is a $N \times N$ matrix for $i=1,\dots,p$ and $K$ is the number of terms in the sum  needed to specify $\matA$. Since the quadratic form remains the same by $\vecx^T \matA^{\alpha}_i \vecx = \vecx^T (\matA^{\alpha}_i + \matA^{\alpha,T}_i )\vecx /2$ we can assume without loss of generality that the $\matA^{\alpha}_i$ are symmetric and also that $\opA$ is symmetric. Note that the representation of Eq.~(\ref{Eq:dectensA}) is useful to simplify the computation of the gradient and the Hessian matrix of $f(\vecx)$. However, our quantum algorithm does not explicitly require this tensor decomposition of $\opA$, which may be hard to compute in general. We only require access to the matrix elements of $\opA$ for the use in standard Hamiltonian simulation methods~\cite{Berry2012} (see the data input discussion below).
In this work, we consider general but sparse matrices $\matA$, where the sparsity requirement arises from the quantum simulation methods. A sparse matrix $\matA$ represents a polynomial with a relatively small number of monomials.

For the gradient descent and Newton's methods, we need to compute expressions for the gradient and the Hessian matrix of the objective function. In the tensor formulation of Eq.~(\ref{Eq:dectensA}), the gradient of the objective function at point $\vecx$ can be written as
\begin{equation}
\nabla f(\vecx) =  \sum_{\alpha=1}^K \sum_{j=1}^p   \left (\prod^p_{i=1 \atop i\neq j}\vecx^{T} \matA^{\alpha}_i \vecx \right) \; \matA^{\alpha}_j \; \vecx =: \opD(\vecx) \vecx,
\label{Eq:deriv}
\end{equation}
which can be interpreted as an operator $\opD \equiv \opD(\vecx)$ which is a sum of matrices $\matA^{\alpha}_j$ with $\vecx$-dependent coefficients $\prod_{i\neq j}\vecx^{T} \matA^{\alpha}_i \vecx$, applied to a single vector $\vecx$. Note that of course the ordering of the scalar factors in the product $\prod^p_{i=1 \atop i\neq j}\vecx^{T} \matA^{\alpha}_i \vecx$ is unimportant.
In addition, for the matrix norm of the operator $\opD$, note that with $\vert \vecx \vert \leq 1$ we have $\vecx^T \opD(\vecx) \vecx \leq \max_{\vert \vecx \vert \leq 1} \vecx^T \opD(\vecx) \vecx \leq p \max_{\vert \vecx \vert \leq 1}  (\vecx \otimes \dots \otimes \vecx)^T \opA \vecx \otimes \dots \otimes \vecx \leq p \Lambda_A$. Thus we can define a bound for the norm of $\opD$ with $\Lambda_D>0$ for which $\Vert \opD \Vert \leq \Lambda_D$ and $\Lambda_D = \Ord{p\Lambda_A}$.

In a similar fashion, the Hessian matrix at the same point reads,
\begin{eqnarray}
\matH( f(\vecx))  &=&  \sum_{\alpha=1}^K \sum_{{j, k=1\atop j\neq k}}^p  \prod^p_{{i =1 \atop i \neq j,k}} \big(\vecx^{T} \matA^{\alpha}_i \vecx \big)  \; \matA^{\alpha}_k \vecx\; \vecx^T \matA^{\alpha}_j  + \opD=:\matH_1+ \opD,
\label{Eq:hess}
\end{eqnarray}
defining the operator $\matH_1$. 
For the matrix norm of the operator $\opH$, note that with $\vert \vecx \vert \leq 1$ we have $\vecx^T \opH(\vecx) \vecx \leq \max_{\vert \vecx \vert \leq 1} \vecx^T \opH(\vecx) \vecx \leq p^2 \max_{\vert \vecx \vert \leq 1}  (\vecx \otimes \dots \otimes \vecx)^T \opA \vecx \otimes \dots \otimes \vecx \leq p^2 \Lambda_A$. Thus we can define a bound for the norm of $\opH$ with $\Lambda_H>0$ for which $\Vert \opH \Vert \leq \Lambda_H$ and $\Lambda_H = \Ord{p^2\Lambda_A}$.
The core of the quantum algorithms for gradient descent and Newton's method will be to implement these matrices as quantum operators acting on the current state and successively shifting the current state towards the desired minimum.

\subsection{Spherical constraints}

\begin{figure}
\centering
\includegraphics[scale = 0.5]{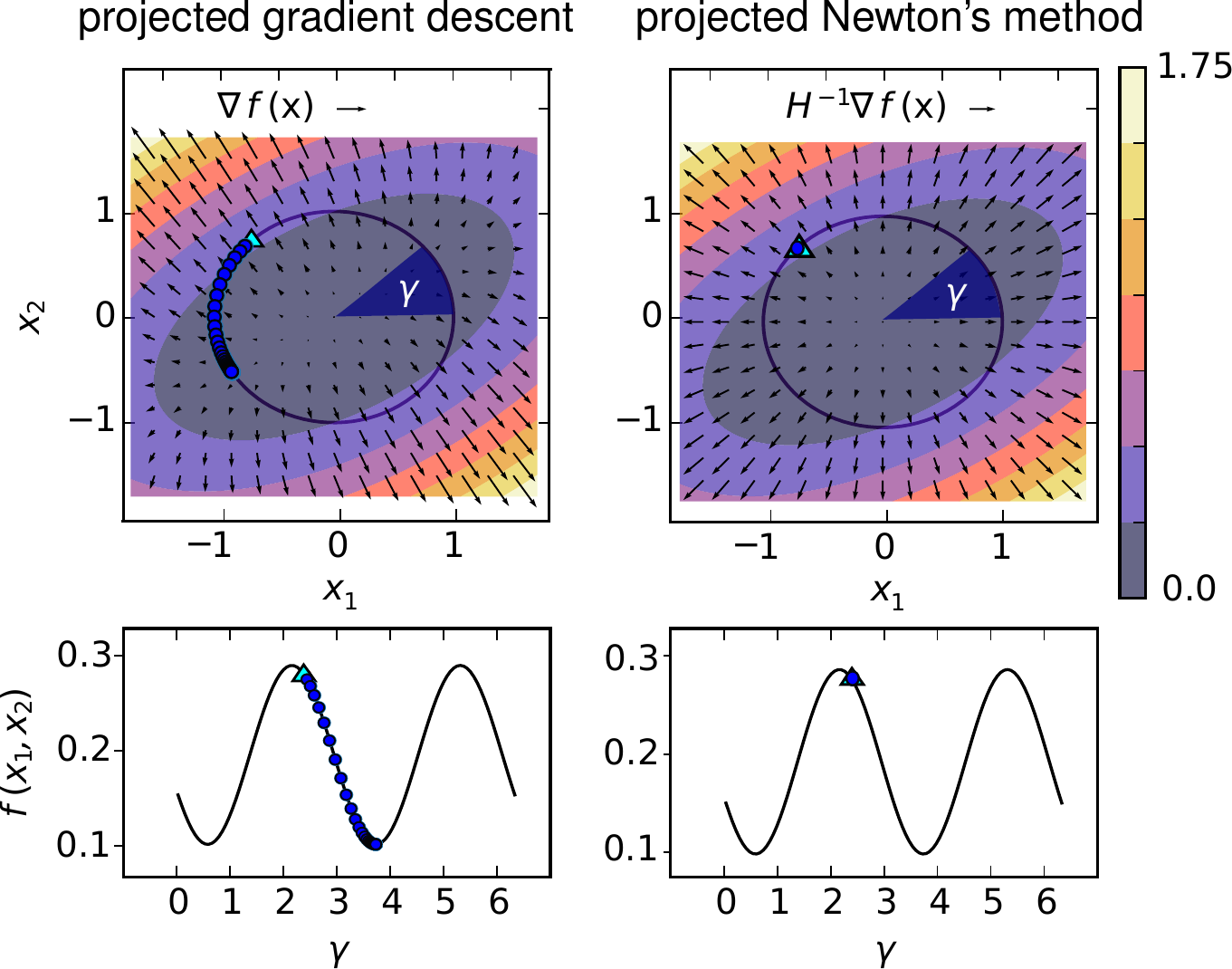}
\caption{Projected gradient descent and projected Newton's method for quadratic optimization under unit sphere constraints (i.e., the solution is constrained to the circle). Below is the function on the feasible set only, with the angle $\gamma$ starting from position $(1,0)^T$. The parameters are chosen as $K = 1, p = 1, N = 2$ and the objective function is $f(\vecx) = \vecx^T \matA \vecx$ with $a_{11} = 0.3, a_{12} =a_{21}= -0.2,  a_{22} = 0.5$ and the initial point (light blue triangle) is chosen as $\vecx^{(0)} = [-0.707, 0.707]$. For quadratic forms, Newton's method struggles to find the minimum on the feasible set plotted below, since the field lines of the descent direction $\matH^{-1} \nabla f(\vecx)$ are perpendicular to the unit circle. This also holds for matrices $A$ that are not positive definite as in this example. }\label{Fig:example_saddle}
\end{figure}

\begin{figure}
\centering
\includegraphics[scale = 0.5]{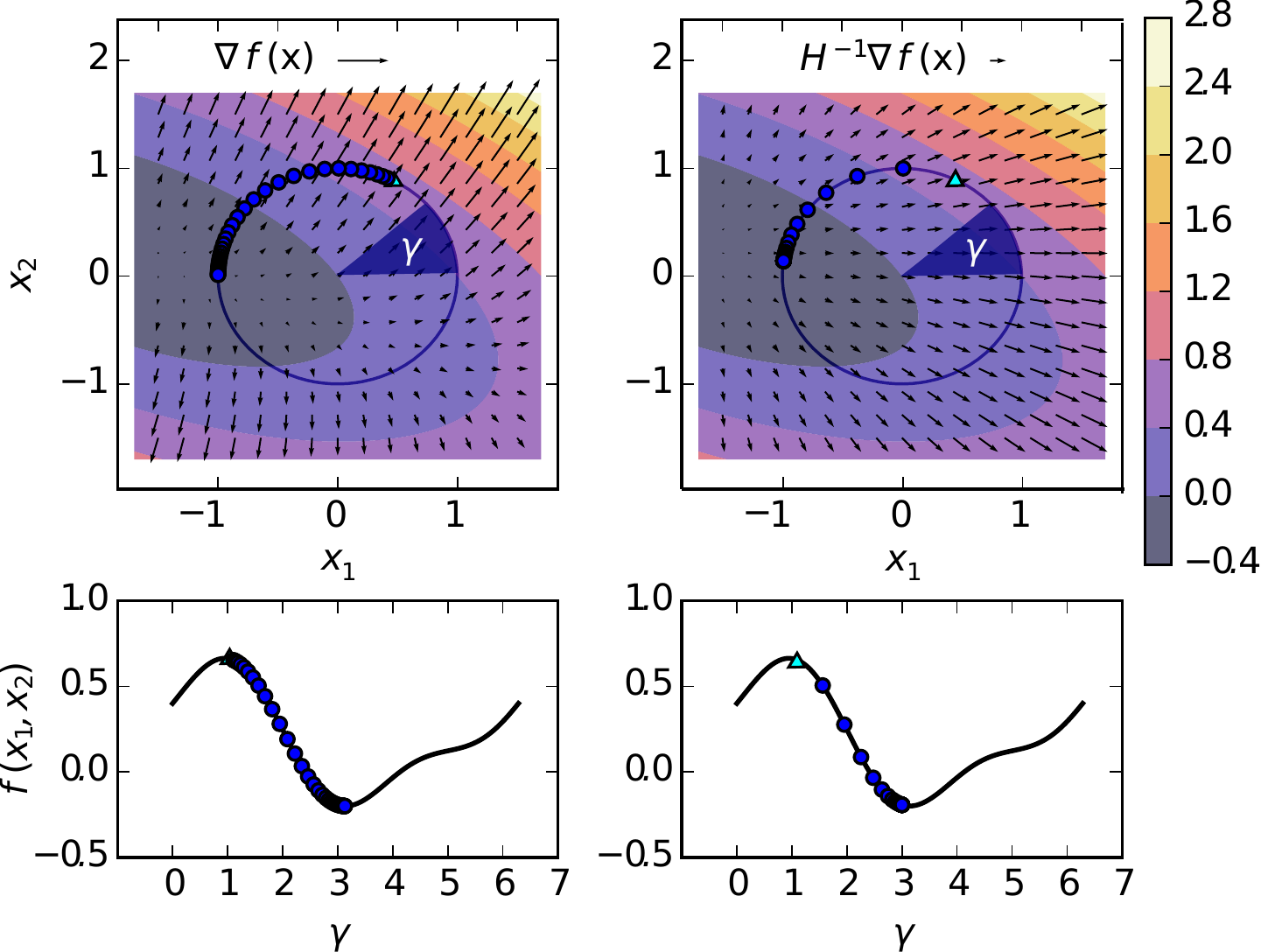}
\caption{The quantum algorithms can be adapted to optimize polynomials that include inhomogeneities. Here, we consider a quadratic form, where the parameters including the inhomogeneity are the same as in Figure \ref{Fig:standard_descent}. The panels are analogous to Figure \ref{Fig:example_saddle}. Newton's method arrives at the solution faster than gradient descent.
}
\label{Fig:example_inhom}
\end{figure}

Since in this work we represent vectors $\vecx$ as quantum states, the quantum algorithm naturally produces normalized vectors with $\vecx^T\vecx = 1 $, thereby implementing a constraint known in the optimization literature as a \textit{spherical constraint}. Applications of such optimization problems appear in image and signal processing, biomedical engineering, speech recognition and quantum mechanics \cite{he10}.

In addition, we include further standard assumptions \cite{Nocedal2006}. We assume an initial guess $\vecx_0$ reasonably close to the constraint local minimum. We assume that the Hessian of the polynomial in Eq.~(\ref{Eq:hess}) in the vicinity of the solution is positive semidefinite. 
We also include an assumption about the smoothness/continuity of the polynomial, which is straightforwardly satisfied as long as the polynomial does not diverge unreasonably in the optimization region under consideration. These assumptions  guarantee saddle-point free optimization and monotonic convergence to the minimum.

The problem we attempt to solve is therefore defined as follows:
\begin{problem} \label{problem1} Let $f(x)$ be a homogeneous polynomial of even order $2p\in \mathbbm Z>0$ as in Eq.~(\ref{Eq:of2}).
Let the matrix $\opA$ defining $f(x)$ be symmetric and have sparsity $s_A$, defined as the number of non-zero matrix elements in each row and column.
In addition, let $\Lambda_A >0$ be given such that the matrix norm of $\opA$ is $\Vert \opA \Vert \leq \Lambda_A$, and without loss of generality $\Lambda_A \geq 1$.
Starting with an initial guess $\vecx_0$, solve the problem
\[
\min_{\vecx} f(\vecx),
\]
subject to the constraint $\vecx^T\vecx = 1$ by finding a local minimum $ \vecx^\ast$.
We assume that the Hessian is positive semidefinite at the solution, i.e., $\opH(f(\vecx^\ast)) \geq 0$, and that the initial guess is sufficiently close to the solution. In addition, the polynomial is smooth and we assume that it is Lipschitz continuous in 
the region of optimization. 
\end{problem}
A well-known adaptation of gradient descent for such constrained problems is \textit{projected gradient descent} \cite{Goldstein64,Levitin66}, where after each iteration the current solution is projected onto the feasible set (here corresponding to a renormalization to $\vecx^T \vecx = 1$). It turns out that our quantum algorithms naturally follow this procedure and the renormalization is realized in each update of the quantum state. We therefore obtain the general convergence properties of the projected gradient descent and projected Newton's methods. Note that for the simple case of the quadratic function $\vecx^T \matA \vecx$, i.e., for $p=1$, gradient descent with a unit norm constraint finds one of the eigenstates of $\opA$.

Although the choice of the objective function allows for an elegant implementation of the operators $\opD$ and $\opH^{-1}$ by means of quantum information processing, it is in some cases not suited for Newton's method. For example, if $p=K=1$ the objective function reduces to a quadratic form and $\opH^{-1} \nabla f(\vecx) = \opD^{-1} \opD \vecx = \vecx$. The direction of search is consequently perpendicular to the unit sphere and Newton's method does not update the initial guess at all (see Figure \ref{Fig:example_saddle}).
For this reason and, more generally, to increase the class of functions that can be optimized, it is interesting to include inhomogeneous terms to the polynomial. For example, a simple linear inhomogeneity can be added by considering the polynomial function $f(\vecx) + \mathbf{c}^T \vecx$, where $f(\vecx)$ is the homogeneous part as before and $ \mathbf{c}$ is a vector specifying the inhomogeneous part (see Figure \ref{Fig:example_inhom}).
We will sketch a method to include such inhomogeneities in Section~\ref{sec:inhpoly}.

\section{Quantum gradient descent algorithm}

\subsection{Data input model}\label{Sec:inputmodel}
To implement a quantum version of the gradient descent algorithm we assume without loss of generality that the dimension $N$ of the vector $\vecx$ is $N=2^{n}$ where $n$ is an integer.
A vector not satisfying this condition can always be padded to $2^{n}$ as long as the additional coordinates do not play a role for the optimization through matrix $\matA$. As mentioned, we consider the case of a normalization constraint $\vecx^T\vecx = 1$. Following previous quantum algorithms \cite{Harrow09,Rebentrost14, Lloyd14} we represent the entries of $\vecx = (x_1,\dots,x_N)^T$ as the amplitudes of a $n$-qubit quantum state  $\ket{\vecx} = \sum_{j=1}^N x_j \ket{j} $, where $\ket{j}$ is the $j$'th computational basis state.

In this work, we assume the following oracles for the data input. First, we define the oracle for providing copies of the initial state.
Let $\vecx^{(0)} = (x^{(0)}_1,\dots,x^{(0)}_N)^T$ be the initial point we choose as a candidate for finding the minimum $\vecx^*$ with $\sum_i |x^{(0)}_i|^2=1$.
\begin{oracle}[Initial state preparation oracle]\label{oracleInit}
There exists an oracle that performs the operation $ \ket {0} \to \ket{\vecx^{(0)}}$ on $n$ qubits. 
\end{oracle}
We assume that the initial quantum state $\ket{\statex^{(0)}}$ corresponding to
$\vecx^{(0)}$ via amplitude encoding can be prepared efficiently, either as an output of a quantum algorithm or by preparing the state from quantum random access memory \cite{Giovannetti2008,Giovannetti2008_2,Martini2009}. For example, efficient algorithms exist for preparing states corresponding to integrable \cite{Grover2002}  and bounded \cite{Soklakov06}  distributions. In addition, we assume an oracle for the matrix $\opA$ \cite{Berry2012}.

\begin{oracle} [Polynomial coefficients oracle]\label{oracleA}
  Let $j,k=1,\cdots,N^p$ with $N = 2^n$. There exists an oracle that performs the operation $\ket{j,k} \ket 0 \to \ket{j,k} \ket {A_{jk}}$ on $2p n + n_\chi$ qubits where $A_{j k}$ is encoded to accuracy $\chi = 2^{-n_\chi}$.
\end{oracle}
We assume that the error $\chi$ is much smaller than other errors and does not affect the analysis \cite{Berry2012}.
Note that the indices $j,k$ can isomorphically be described by numbers $i_1,\cdots,i_{2p}=1,\cdots, N$. Thus, the oracle can be equivalently given as
$\ket{i_1,\cdots,i_{2p}} \ket 0 \to \ket{i_1,\cdots,i_{2p}} \ket {A_{i_1,\cdots,i_{2p}}}$.
To take advantage of  sparsity, we also assume the following oracle which allows us to choose the non-zero matrix elements.
\begin{oracle} [Sparse input oracle] \label{oracleSparse} Let $j=1,\cdots,N^p$ and $l=1,\cdots,s_{A}$.
  There exists an oracle that performs the operation  $ \ket{j,l} \to \ket{j, g_A(j,l)}$ on $2p n$ qubits where the efficiently computable function $g_A(j,l)$ gives the column index of the $l$-th nonzero element of row $j$ of matrix $\opA$.
\end{oracle}
These Oracles (\ref{oracleA}) and (\ref{oracleSparse}) allow for an efficient simulation of $\eu^{-\im \opA t}$ via the methods in \cite{Berry2012} and are often a standard assumption in the literature.
In the remainder of this section, we will first describe how to quantumly simulate the gradient operator as $e^{- i \opD t}$ and how to implement
a quantum state $\ket{\nabla f(\vecx)}$ representing the gradient $\nabla f(\vecx)$.
We then describe how to obtain the update of the current candidate $\ket{\vecx^{(t)}}$ at the $t$-th step of the gradient descent method via the current-step gradient quantum state. Finally, we discuss  the run time of multiple steps.

\subsection{Computing the gradient}\label{subGradient}

In this section, we present the quantum simulation of the gradient operator $\opD$ and the preparation of the gradient state $\opD \ket{\vecx}$, which may be of independent interest. For the definition of $\opD$, see Eq.~(\ref{Eq:deriv}).
We omit the index $t$ indicating the current step for readability. For the diagonalizable matrices discussed here, it sufficies to use the operator norm $\Vert \cdot \Vert = \max_j \vert \lambda_j(\cdot) \vert$, where $\lambda_j$ are the eigenvalues of the matrix under consideration.

\begin{result}[Quantum simulation of the gradient operator]
Given the exact quantum states $\ket \vecx$ as well as Oracles (\ref{oracleA}) and (\ref{oracleSparse}). Let the gradient operator corresponding to the state $\ket \vecx$ be $\opD= \sum_{\alpha}   \sum_{j} \left(\prod_{i\neq j}\vecx^{T} \matA^{\alpha}_i \vecx \right)  \opA_{j}^{\alpha}$ with $\Vert \opD \Vert \leq \Lambda_D$. Then there exists a quantum algorithm that simulates $\eu^{-\im \opD \tau}$ to accuracy $\epsilon$ using $ \Ord{\frac{p^3 \Lambda_D^2 \tau^2}{\epsilon} }$ copies of $\ket{\vecx}$, 
$\Ord{\frac{p^2 \Lambda_D^3 \tau^2 s_A }{\epsilon} }$ queries to the oracles for $\opA$, and $\tOrd{\frac{p^3 \Lambda_D^3 \tau^2 s_A}{\epsilon} \log N } $ single- and two-qubit gates.
\end{result}

First, note that the classical scalar coefficients $\vecx^{T} \matA^{\alpha}_i \vecx$ that occur as weighing factors in Equations (\ref{Eq:deriv}) and (\ref{Eq:hess}) can be written as the expectation values of operators $\opA^{\alpha}_i$ and quantum states $\ket{\statex}$,  or  $\bra{\statex}\opA^{\alpha}_i\ket{\statex} =  \mathrm{tr} \{ \opA^{\alpha}_i \rho \}$  where $\rho = \ketbra{\statex}{\statex}$ is the corresponding density matrix. With this relation, we show that the gradient operator can be implemented as a Hamiltonian in a quantum algorithm.  The gradient operator $\opD$ can be represented as
\begin{equation}\label{eqDtoMD}
\opD =  \mathrm{tr}_{1\dots p-1} \left \{\left( \rho^{\otimes (p-1)} \otimes \mathcal{I} \right) \;  \opMD \right \},
\label{Eq:proxy}
\end{equation}
where $\rho^{\otimes (p-1)} = \rho  \otimes \cdots \otimes \rho$ is the joint quantum state of $p-1$ copies of $\rho$. This operator $\opD$ acts on another copy of $\rho$. The auxiliary operator $\opMD$ is independent  of $\rho$ and given by
\begin{equation} \label{eqMD}
 \opMD = \sum\limits_{\alpha=1}^K \sum\limits^p_{j=1} \left( \bigotimes^p_{{i=1 \atop i\neq j}} \opA_{i}^{\alpha} \right )\otimes \opA_{j}^{\alpha}.
\end{equation}
More informally stated, the operator $\opMD$ can be represented from the
$\matA^{\alpha}_j,\ j=1 \dots p, $ of Eq.~(\ref{Eq:dectensA}) in such a way that for each term in the sum the expectation values of the first $p-1$ subsystems correspond to the desired weighing factor, and the last subsystem remains as the operator acting on another copy of $\rho$.
Note that the order of the factors in the product $ \bigotimes^p_{{i=1 \atop i\neq j}} \opA_{i}^{\alpha} $ only changes the matrix $\opMD$ but not the operator $\opD$.
The matrix $\opMD$ is a sum over $p$ matrices that have the same sparsity as $\opA$, hence its sparsity is $\Ord{p s_A}$.
The matrix exponential $\eu^{-\im\opMD \tau}$ can be simulated efficiently on a quantum computer with access to the above oracles, as shown in \ref{appendixSim}, Lemma \ref{Lm:AandMD}.
Simulating a time $\tau$ to accuracy $\epsilon$, we require $\Ord{p^3 \Lambda_A s_A \tau^2/\epsilon}$ queries to the oracles for $\opA$ and $\tOrd{p^4 \Lambda_A s_A \tau^2 \log N/\epsilon}$ quantum gates.

With this mapping from $\opD$ to $\opMD$ and the exponentiation of $\opMD$, we can exponentiate $\opD$. The idea is to implement the matrix exponentiation $\eu^{-\im \opD \Delta t}$ adapting the quantum principal component analysis (QPCA) procedure outlined in \cite{Lloyd14}. Since $\opD$ depends on the current state $\ket{\statex}$, we cannot simply use the oracular exponentiation methods of \cite{Berry2012}. Instead we exponentiate the sparse matrix $\opMD$ given in Eq.~(\ref{eqMD}).
For a short time $\Delta t$, use multiple copies of $\rho=\ketbra{\statex}{\statex} $ and perform a matrix exponentiation of $\opMD$. In the reduced space of the last copy of $\rho$, we observe the operation
\begin{equation}
\mathrm{tr}_{p-1} \{ \eu^{-\im\opMD\Delta t} \; \rho^{\otimes p}  \; \eu^{\im \opMD\Delta t} \} = \eu^{-\im \opD \Delta t} \rho \; \eu^{\im \opD \Delta t} + \mathcal E.
\label{Eq:pca}
\end{equation}
where
\begin{equation} \label{eqPCAError}
\mathcal E = \mathcal E_{\opMD} + \mathcal E_{\rm samp }.
\end{equation}
Here, the error term $\mathcal E$ contains contributions from the erroneous simulation of $\eu^{-\im\opMD\Delta t}$ given by $\mathcal E_{\opMD} $ and the intrinsic error of the sample-based Hamiltonian simulation given by $\mathcal E_{\rm samp } $. We use the operator norm $\Vert \cdot \Vert$ to bound the error terms.

Regarding the error $\mathcal E_{\opMD}$, we can choose $\Vert \mathcal E_{\opMD} \Vert = \Ord{p^2 \Delta t^2}$, using $\Ord{p \Lambda_A  s_A }$ queries to the oracle for $\opA$ and using $\tOrd{p^2 \Lambda_A  s_A \log N}$ gates,
see \ref{appendixSim}, Lemma \ref{Lm:AandMD}.
Regarding the error $\mathcal E_{\rm samp }$, its size is given by $\Vert\mathcal E_{\rm samp } \Vert = \Ord{\Lambda_D^2 \Delta t^2 }$ similar to \cite{Lloyd14}. The total error for a small time step $\Delta t$ can be upper bounded by $\Vert\mathcal E \Vert = \Ord{p^2 \Lambda_D^2 \Delta t^2 }$.
For  a total time $\tau= n \Delta t$ and given accuracy $\epsilon$, we have $\epsilon = n \Vert\mathcal E \Vert =  \Ord{\frac{p^2 \Lambda_D^2 \tau^2}{n} }$. Thus $n =  \Ord{\frac{p^2 \Lambda_D^2 \tau^2}{\epsilon} }$
steps are required and for each step we use $\Ord{p}$ copies of $\rho$. Thus we need a total of $ \Ord{\frac{p^3 \Lambda_D^2 \tau^2}{\epsilon} }$ copies of $\ket{\vecx}$. We require $\Ord{n p \Lambda_A  s_A } =\Ord{\frac{p^3 \Lambda_D^2 \tau^2 \Lambda_A s_A }{\epsilon} }
=\Ord{\frac{p^2 \Lambda_D^3 \tau^2 s_A }{\epsilon} }$ queries to the oracles for $\opA$ and $\tOrd{n p^2 \Lambda_A  s_A \log N} = \tOrd{\frac{p^4 \Lambda_D^2 \tau^2 \Lambda_A  s_A}{\epsilon} \log N } = \tOrd{\frac{p^3 \Lambda_D^3 \tau^2 s_A}{\epsilon} \log N } $ gates, using that $\Lambda_A = \Ord{\Lambda_D/p}$.

This results assumes perfect states $\ket \vecx$. However, within the gradient descent routine, the states pick up errors due to the imprecise simulation and phase estimation methods. We discuss a generic result  regarding the simulation with such imperfect states using a particular error model. The error model assumes identically, independently distributed (i.i.d.) random error vectors. Such random vectors result in scalar-valued random variables when used to form scalar-valued quantities such as $\bra \vecx \matA_j^\alpha \vec v $ of the matrices $\opA^{\alpha}_j$. One can apply the central limit theorem for a sum of such random variables. Such i.i.d.~models are generally a natural start for an error discussion. The context of the gradient descent method begs the question how applicable such a model is because it is reasonable that the errors do not arise in an independent fashion but rather in a somewhat correlated sense as each step depends on the previous steps. We note that central limit theorems have been shown also for correlated random variables \cite{Hoeffding1994}.
In addition, we note that certain errors are tolerable. For example, in stochastic gradient descent, a gradient is computed from randomly sampling the training set. This leads to surprisingly good learning behavior for neural networks \cite{Hinton06,Bengio09}. More generally it has even been shown that the stochasticity in learning problems leads to better generalization ability through an indirect regularization of the problem.
In many practical situations, sufficiently close to the minimum of a convex optimization problem, a small error in the gradient can nevertheless lead to convergence if the step size is chosen appropriately.
A more general and quantitative analysis of correlated errors can give further insight into the algorithms behavior and is left for future work.

\begin{assume}[I.i.d.~error model] \label{assumeError}
Let the quantum states $\ket {\tilde \vecx}$ be given
with independent, bounded, random errors, i.e.~$  \ket{\tilde \vecx} - \ket{ \vecx}= \vec v$, with $\vec v$ a random vector 
  such that $\vert  \vec v \vert  \ll 1 $.
\end{assume}
Under this error model the simulation performance decreases as $1/\epsilon \to 1/\epsilon^2$.
\begin{result}[Quantum simulation of the gradient operator with i.i.d. imperfect states] \label{resultSimDError}
Assume Oracles (\ref{oracleA}) and (\ref{oracleSparse}).
Let the quantum states $\ket { \vecx}$ be given with errors according to Assumption \ref{assumeError}.
 Let the gradient operator corresponding to the state $\ket \vecx$ be $\opD= \sum_{\alpha}   \sum_{j} \big(\prod_{i\neq j}\vecx^{T} \matA^{\alpha}_i \vecx \big)  \opA_{j}^{\alpha}$ with  $\Vert \opD\Vert \leq \Lambda_D$. Then there exists a quantum algorithm that simulates $\eu^{-\im \opD \tau}$ to accuracy $\epsilon$ using
$ \Ord{p^5 \Lambda_D^2  \frac{\tau^2}{\epsilon^{2}}  } $ copies of $\ket{\vecx}$,
$\Ord{\frac{p^4 \Lambda_D^3 \tau^2 s_A }{\epsilon^2} }$ queries to the oracles for $\opA$ and $\tOrd{\frac{p^5 \Lambda_D^3 \tau^2 s_A}{\epsilon^2} \log N } $ quantum gates.
\end{result}
The modification here is that the error terms for a simulation of time $\Delta t$ becomes
\begin{equation} \label{eqPCAErrorState}
\mathcal E' = \mathcal E_{\opMD} + \mathcal E_{\rm samp } + \mathcal E_{\ket \vecx },
\end{equation}
in comparison to Eq.~(\ref{eqPCAError}). The additional error term is $ \mathcal E_{\ket \vecx }$.
An error in $\ket \vecx $ of size $<1$ directly translates into the error $ \Vert \mathcal E_{\ket \vecx } \Vert = \Ord{p \Delta t}$ linear in $\Delta t$. However, the norm $ \Vert \mathcal E_{\ket \vecx } \Vert $ omits the sign of the error term.
For multiple steps of size $\Delta t$ the error terms add up with their respective signs. Under the assumed error model, such addition of errors leads to the central limit theorem and the error averages to certain extend. We refer to \ref{App:SamBasedHam}, Lemma \ref{lemmaSampleD}, for further discussion.

The next result shows how to prepare a quantum state $\opD \ket{\vecx}$. 
\begin{result} [Gradient state preparation with imperfect states]\label{resultGradient}
Given Oracles (\ref{oracleA}) and (\ref{oracleSparse}). Let quantum states $\ket \vecx$ be given according to the error model of Assumption \ref{assumeError} with error bound $\Ord{ \epsilon_D}\ll1$. Let the gradient operator $\opD$ corresponding to the state $\ket \vecx$ have condition number $\kappa_D$. Then there exists a quantum algorithm that applies the gradient operator to the state $\ket \vecx$. That is, a state $\propto \opD \ket \vecx$ can be prepared to accuracy $\Ord{ \epsilon_D}$ with the required resources given by setting $\epsilon\to \epsilon_D$ and $\tau \to \kappa_D/\epsilon_D$ in Result \ref{resultSimDError} and $\bigO(\lceil \log (2+1/2\epsilon_D) \rceil )$ additional ancilla qubits. 
We require
$\Ord{\kappa_D^2}$
repetitions to success with high probability.
\end{result}

We can implement the multiplication $\opD \ket{\statex} = \ket{\nabla f(\statex)}$ via phase estimation. In order to perform phase estimation and extract the eigenvalues of $\opD$, we need to implement a series of powers of the exponentials in Equation (\ref{Eq:pca}) controlled on a time register. The infinitesimal application of the $\opMD$ operator can be modified as in Refs.~\cite{Lloyd14,Kimmel2017},
as $\ketbra{0}{0} \otimes \mathcal{I} + \ketbra{1}{1}  \otimes \eu^{-\im \opMD \Delta t}$ and applied to a state $\ketbra{q}{q}\otimes \rho $ where $\ket q$ is an arbitrary single control qubit state. For the phase estimation, use
a multi-qubit register with $\bigO(\lceil \log (2+1/2\epsilon_D) \rceil )$ control qubits forming an eigenvalue register to resolve eigenvalues to accuracy $\epsilon_D$.
Then apply  $\sum_{l=1}^{S_{\rm ph} }\ketbra{l\Delta t}{l \Delta t} (\eu^{-\im \opD \Delta t})^l  $ for $S_{\rm ph}$ steps.
If $S_{\rm ph}$ is chosen appropriately, see below, this conditioned application of the matrix exponentials
allows to prepare a quantum state proportional to
\begin{equation}
\sum_j  \beta_j  \ket{u_j(\opD)} \ket{\tilde \lambda_j(\opD) }.
\end{equation}
Here, $\ket{\statex} =\sum_j \beta_j \ket{u_j(\opD)} $ is the original state $\ket{\statex}$ written in the eigenbasis $\{\ket{u(\opD)_j}\}$ of $\opD$, and $\ket{\tilde \lambda(\opD)_j }$ is the additional register encoding the corresponding eigenvalue $\lambda(\opD)_j$ in binary representation with accuracy $\epsilon_D$.
 As in \cite{Harrow09,Wiebe12}, we conditionally rotate an extra ancilla qubit, uncompute the phase estimation, and measure the ancilla qubit, which results in
\begin{equation}
\sum_j \tilde \lambda_j(\opD)  \beta_j  \ket{u_j(\opD)}.
\end{equation}
This performs the matrix multiplication with $\opD$. The accuracy
of this matrix multiplication is $\epsilon_D$ if the total time of the phase estimation is taken to be $S_{\rm ph} \Delta t=\bigO(\kappa_D/\epsilon_D)$ \cite{Wiebe12}.
As the main algorithmic step is controlled Hamiltonian simulation, we obtain the resource requirements from Result \ref{resultSimDError} by replacing $\epsilon \to \epsilon_D$ and 
$\tau\to \kappa_D/\epsilon_D$.
In addition, the ancilla measurement success probability is $1/\kappa_D^2$ \cite{Wiebe12}. Thus, we require $\Ord{\kappa_D^2}$ repetitions for a high probability of success,
which completes Result~\ref{resultGradient}.

\subsection{Single gradient descent step}
\label{Sec:singlestep}
In this section, we establish the gradient descent procedure. To decouple the simulation method with the main 
result, we use the following assumption for the simulation method of $\opD^{(t)}$. 
\begin{assume}\label{assumeSimD}
There exists a quantum simulation method for the operator $\opD^{(t)}$ with $\Vert \opD^{(t)} \Vert \leq \Lambda_D$. We can simulate the controlled operator $\eu^{-\im \ket 1 \bra 1 \otimes \opD^{(t)} \tau}$ to accuracy $\epsilon_D$ using $\Ord{{\rm poly} (p,s_A,\Lambda_D,\tau, 1/\epsilon_D)}$ copies of the state $\ket{ \vecx^{(t)}}$ and queries to the oracles for $\opA$. In addition, the method requires 
$\tOrd{{\rm poly} (p,s_A,\Lambda_D) \log N}$ basic quantum gates for each copy.
\end{assume}
Given the simulation method, we can perform a single gradient descent step.
\begin{result}[Single gradient descent step] \label{resultMain}
 Given quantum states $\ket{ \statex^{(t)}}$ encoding the current solution at time step $t$ to Problem \ref{problem1} in amplitude encoding to accuracy $\epsilon^{(t)}>0$, as well as ancilla qubits.
Let the gradient operator corresponding to the state
$\ket{ \statex^{(t)}}$ be given by $\opD^{(t)}$
with $ \Vert \opD^{(t)}\Vert \leq \Lambda_D$.
Assume a quantum algorithm for the simulation of $\opD^{(t)}$ according to Assumption \ref{assumeSimD} exists. 
Let a step size be given by $0<\eta^{(t)}<1/(2 \Lambda_D)$.
Then there exists a quantum algorithm using Oracles (2) and (3) such that a single step of the gradient descent method of step size $\eta^{(t)}$ to prepare an improved normalized quantum state $\ket{\vecx^{(t+1)}}\propto
\left( \ket{\statex^{(t)}} - \eta^{(t)} \opD^{(t)} \ket{\statex^{(t)}} \right)$ to accuracy $\epsilon^{(t+1)}= \bigO(\eta^{(t)} \epsilon_D^{(t)}+\epsilon^{(t)})$  can be performed. This quantum algorithm requires
$\Ord{{\rm poly} (p,s_A,\Lambda_D,1/\epsilon_D)}$ 
copies of the state $\ket{ \vecx^{(t)}}$ and queries to the oracle of $\opA$, and $\tOrd{{\rm poly} (p,s_A,\Lambda_D,1/\epsilon_D) \log N}$  basic quantum gates.
The success probability of this algorithm is lower-bounded by $1/16$.
\end{result}
We shorthand $\opD := \opD^{(t)}$ in the following. Starting with the current solution, prepare the state
\begin{equation}
\ket{\statex^{(t)}}\left(\cos\theta \ket{0}_{g}  - \im \sin \theta \ket{1}_{ g} \right ) \ket {0\dots0},  \label{eq:prep}
\end{equation}
where $\theta$ is an external parameter determining the gradient step size, the first register contains the state $\ket{\statex^{(t)}}$,  the second register contains a single ancilla qubit used for the gradient step (subscript $g$), and the final register contains ancilla qubits for the gradient matrix multiplication.
The operator $\opD$ can be multiplied to $\ket{\statex^{(t)}}$ conditioned on the first ancilla being in state $\ket{1}_g$.
Let the eigenstates of $\opD$ be given by $\ket{u_j(\opD)}$ and the eigenvalues by $\lambda_j(\opD)$.
After the conditional phase estimation, we obtain:
\begin{equation}
 \ket{ \psi} =\left( \cos\theta\ket{\vecx}  \ket{0}_g \ket{0\dots0} +\im \sin \theta  \sum_j \beta_j \ket{u_j(\opD)}\ket{1}_g \ket{\tilde \lambda_j(\opD)}\right ),
\end{equation}
where $\beta_j = \langle u_j(\opD) \ket{\vecx}$ and $\tilde \lambda_j(\opD)$ is an approximation to $\lambda_j(\opD)$ with accuracy $\epsilon_D$ Now use another ancilla (subscript $d$.) For the $\ket 0_g$ path this ancilla is initialized in the $\ket 1_d$ state and for the $\ket 1_g$ path perform a conditional rotation of this ancilla. Uncomputing the eigenvalue register arrives at the state
\begin{eqnarray}
 \cos\theta \ket{\vecx} \ket{0}_g \ket{1}_d + \nonumber \\ \im \sin \theta \sum_j \beta_j \ket{u_j(\opD)}\ket{1}_g \left(\sqrt{1-\left(\xi_D \tilde\lambda_j(\opD) \right)^2}\ket{0}_d + \xi_D \tilde\lambda_j(\opD) \ket{1}_d \right).
\end{eqnarray}
Chose a constant $\xi_D$ such that $1/\Lambda_D > \xi_D \geq \eta$, such that the rotation is well defined, where $\Lambda_D \geq \max_j \vert \lambda_j(\opD) \vert$.
Now perform a joint measurement of the final two ancillas. The first of these ancillas is required for the vector addition of the current solution and the state related to the first derivative, the second for the matrix multiplication.
The first ancilla is measured in the basis $\ket{\mathrm{yes}} = \frac{1}{\sqrt{2}}(\ket{0}_g + i \ket{1}_g)$ and $\ket{\mathrm{no}} = \frac{1}{\sqrt{2}} ( i\ket{0}_g + \ket{1}_g)$. The second ancilla is measured in the state $\ket 1_d$. Thus, if $\ket{\mathrm{yes}} \ket 1_d$ is measured, we arrive at (omitting normalization)
\begin{eqnarray}
&\propto &\left ( \cos \theta \ket{\vecx^{(t)}} - \xi_D \sin \theta\ \sum_j \tilde \lambda_j(\opD)  \beta_j \ket{u_j(\opD)} \right) \nonumber \\ &\equiv&  \left ( \cos \theta \ket{\vecx^{(t)}} - \xi_D \sin \theta\ \opD \ket{\vecx^{(t)}} \right).
\end{eqnarray}
Choose $\theta$ such that
\begin{eqnarray}
\cos \theta &=&\frac{1}{\sqrt{1+ \eta^2/ \xi_D^2}},\
\sin \theta = \frac{\eta}{\xi_D \sqrt{1+\eta^2/ \xi_D^2}},
\label{Eq:thetachoice}
\end{eqnarray}
to obtain the state
\begin{equation}
\ket{\vecx^{(t+1)}}=\frac{1}{C_D^{(t+1)}}\left( \ket{\statex^{(t)}} - \eta \ket{\nabla f(\statex^{(t)})} \right),
\label{Eq:updatedstate}
\end{equation}
with
\begin{equation}
\left(C_D^{(t+1)} \right)^2= 1-2\eta \bra{ \vecx^{(t)}} \opD \ket{ \vecx^{(t)}} +\eta^2
\bra{ \vecx^{(t)}} \opD^2\ket{ \vecx^{(t)}} .
\end{equation}
Here, $\langle \vecx^{(t)} | \opD | \vecx^{(t)} \rangle$ refers to the inner product of the current solution $\ket{\vecx^{(t)}}$ with the gradient at this point and  $\sqrt{\bra{\vecx^{(t)}}\opD^2\ket{\vecx^{(t)}}} :=\Delta_D$ is the gradient vector length.
One can express $\left(C_D^{(t+1)} \right)^2 = 1-2\eta \Delta_D \cos \varphi_D +\eta^2 \Delta_D^2$,
 where
$\varphi_D$ is the angle between $\ket{\vecx^{(t)}}$ and $\opD\ket{\vecx^{(t)}}$.
The success probability of the measurement of $\ket{\mathrm{yes}} \ket 1$ is given by
\begin{equation}
P_{{\rm yes},1} = \frac{1}{2(1+\eta^2/\xi_D^2)}\left(C_D^{(t+1)} \right)^2.
\end{equation}
For sufficiently bounded step size, this success probability can be lower bounded.
The angle $\varphi_D = 0$ achieves the lower bound
\begin{eqnarray}
\left(C_D^{(t+1)} \right)^2 \geq \left (1-\eta \Delta_D \right)^2.
\end{eqnarray}
With $0<\eta\leq 1/(2\Lambda_D)$ by assumption and $\Delta_D \leq \Lambda_D$ by definition, we have $\left(C_D^{(t+1)} \right)^2 \geq 1/4$. In addition,  $1+\eta^2/\xi_D^2 \leq 2$ because $\eta \leq \xi_D$ by the choice of $\xi_D$. Thus
$P_{{\rm yes},1} >1/16$.
Thus, the upper bound for the number of repetitions needed, which is $\Ord{1/P_{\mathrm{yes},1}}$, is at most $\Ord{1}$ in this case.

The quantum state $\ket{\vecx^{(t+1)}}$ is an approximation to the classical vector $\vecx^{(t+1)}$, which would be the result of a classical projected gradient descent update departing from a normalized vector $\vecx^{(t)}$. 
The error of the updated step is bounded by the error of the previous step $\epsilon^{(t)}$ plus the error of the conditioned matrix multiplication, i.e.~$\Ord{\epsilon^{(t)}+\eta \epsilon_D^{(t)}}$. Given the assumptions stated above, this state can be prepared via phase estimation, an ancilla rotation and measurement.
The measurement and post-selection normalizes the quantum state at each step. The resource requirements are obtained by setting $\tau \to \Ord{1/\epsilon_D}$ from phase estimation into the simulation method Assumption \ref{assumeSimD}. 
Taking a step size according to the stated assumption, the probability of success and thus the number of repetitions to success are constants.

Result \ref{resultMain} uses generically the Assumption \ref{assumeSimD} for the simulation of $\opD$. 
We now provide a concrete estimate of the resources needed to implement the single-step algorithm using the error model Assumption \ref{assumeError} and our sample-based simulation method Lemma \ref{lemmaSampleD}. 

\begin{result}\label{resultMainResource}
In the setting of Result \ref{resultMain}, include the Assumption \ref{assumeError} for the errors of the current state and a simulation method for the operator $\opD^{(t)}$ according to Lemma \ref{lemmaSampleD}.
The quantum algorithm to perform Result \ref{resultMain} then requires
$\Ord{p^5 \Lambda_D^2/ (\epsilon_D^{(t)})^4    }$
copies of $\ket{ \statex^{(t)}}$. The required number queries to the oracle for $\opA$ is given by
$\Ord{p^4 \Lambda_D^3  s_A/(\epsilon_D^{(t)})^4 }$
and the number of elementary quantum operations is given by
$\tOrd{p^5 \Lambda_D^3 s_A \log N / (\epsilon_D^{(t)})^4}$.
\end{result}
The resource requirements for this result are given from Lemma \ref{lemmaSampleD} for the erroneous sample-based Hamiltonian simulation analyzed in \ref{App:SamBasedHam}. The extension to the controlled simulation as required by phase estimation is done in an overhead that factors in as a constant into the resource requirements \cite{kimmel2017hamiltonian}. Replacing $\tau\to \Ord{1/\epsilon_D^{(t)}}$ gives the stated Result \ref{resultMainResource}.

\subsection{Multiple steps}
\label{sectionMultipleSteps}

\begin{figure}
\centering
\begin{tikzpicture}
\path (0,0) node[] (01) {$\textcolor{mydarkgreen}{\ket{\vecx^{(0)}}}$};
\path (1.0,0) node[ ] (02){$\textcolor{mydarkgreen}{\ket{\vecx^{(0)}}}$};
\path (2.0,0) node[ ] (04) {$\textcolor{mydarkgreen}{\ket{\vecx^{(0)}}}$};
\path (3.25,0) node[ ] (05) {$...$};
\path (4.4,0) node[] (01) {$\textcolor{mydarkgreen}{\ket{\vecx^{(0)}}}$};
\path (5.4,0) node[ ] (02){$\textcolor{mydarkgreen}{\ket{\vecx^{(0)}}}$};
\path (6.4,0) node[ ] (04) {$\textcolor{mydarkgreen}{\ket{\vecx^{(0)}}}$};
\path (7.4,0) node[] (01) {$\textcolor{mydarkgreen}{\ket{\vecx^{(0)}}}$};
\path (8.4,0) node[ ] (02){$\textcolor{mydarkgreen}{\ket{\vecx^{(0)}}}$};
\path (9.4,0) node[ ] (04) {$\textcolor{mydarkgreen}{\ket{\vecx^{(0)}}}$};
\path (10.4,0) node[] {$\textcolor{mydarkgreen}{\ket{\vecx^{(0)}}}$};
\path (11.4,0) node[ ] (05) {$...$};
\path (3.6,-1.1) node[ ] (05) {$...$};
\path (10.0,-1.1) node[ ] (05) {$...$};
\draw [decorate,decoration={brace,amplitude=6pt},xshift=-4pt,yshift=0pt]
(2.5,-0.3) -- (-0.2,-0.3) node [black,midway, below, yshift = -0.4cm] (red)
{ $\textcolor{mydarkred}{\ket{\vecx^{(1)}}}$};
\draw [decorate,decoration={brace,amplitude=6pt},xshift=-4pt,yshift=0pt]
(6.9,-0.3) -- (4.2,-0.3) node [black,midway, below, yshift = -0.4cm]
{ $\textcolor{mydarkgreen}{\ket{\vecx^{(1)}}}$};
\draw [decorate,decoration={brace,amplitude=6pt},xshift=-4pt,yshift=0pt]
(9.7,-0.3) -- (7.0,-0.3) node [black,midway, below, yshift = -0.4cm]
{ $\textcolor{mydarkgreen}{\ket{\vecx^{(1)}}}$};
\draw [decorate,decoration={brace,amplitude=6pt},xshift=-4pt,yshift=0pt]
(6.0,-1.3) -- (3.2,-1.3) node [black,midway, below, yshift = -0.4cm]  (red2)
{ $\textcolor{mydarkred}{\ket{\vecx^{(2)}}}$};
\draw [decorate,decoration={brace,amplitude=6pt},xshift=-4pt,yshift=0pt]
(10.6,-1.3) -- (7.8,-1.3) node [black,midway, below, yshift = -0.4cm]  (green)
{ $\textcolor{mydarkgreen}{\ket{\vecx^{(2)}}}$};
\path (2.1,-2.9) node[anchor=base, fill=white, align=center] (gar){\textcolor{mydarkred}{garbage}};
\path (7.7,-2.9) node[anchor=base, fill=white, align=center] (cor){\textcolor{mydarkgreen}{correct}};
\draw[mydarkred, ->] (gar)--(red);
\draw[mydarkred, ->] (gar)--(red2);
\draw[mydarkgreen, ->] (cor)--(green);
\draw[->] (-0.7,0)--(-0.7,-3);
\path (-1.0,-1.5) node[anchor=base, fill=white, align=center, rotate=90] {iterations};
\end{tikzpicture}
\caption{Both the quantum gradient descent and Newton's method `consume' copies in every step and have only a probability of less than one to produce the correct updated state. This limits their application to local searches with a small amount of iterations and is in principle an underlying feature of iterative quantum routines with probabilistic outcomes.}
\label{Fig:copiesconsumed}
\end{figure}
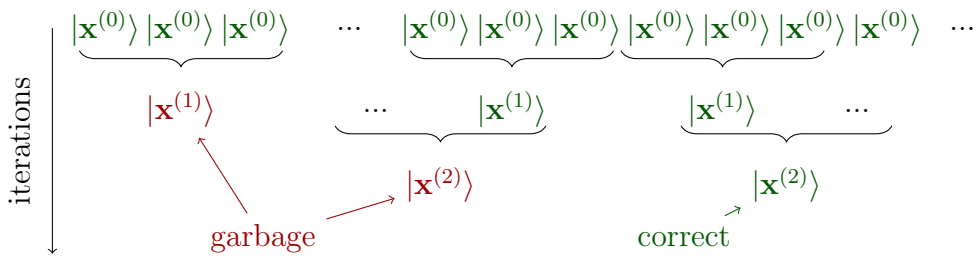

We have presented a single step of the quantum version of  gradient descent for polynomial optimization. We now estimate the resource requirements for multiple steps.
It is reasonable to bound the space explored with the quantum gradient descent algorithm by a constant, as all quantum states live on the surface of a unit sphere.

\begin{result}[Multiple gradient steps]
Assume the setting of optimization of polynomials given in Problem \ref{problem1}.
Let the task be to use the quantum gradient descent method for $T\geq 0$ steps to prepare a
solution $\ket {\vecx^{(T)}}$ to final accuracy $\delta>0$, with step-size schedule $\eta^{(t)}>0$ and  gradient operator norms $\Lambda_D^{(t)}>0$.
Assume that
the space explored with the algorithm is bounded by a constant, i.e.~$T \max_t\eta^{(t)} = \Ord{1}$.
Assume a quantum algorithm for the simulation of $\opD^{(t)}$ according to Assumption \ref{assumeSimD} exists at every step $t$.
There exists a quantum algorithm for this task that requires
 \begin{equation}
 \Ord{{\rm poly} (p,s_A,\Lambda,1/\delta)^T},
\end{equation}
 copies of the initial state $\ket {\vecx^{(0)}}$, where $\Lambda = \max_t  \Lambda_D^{(t)}$. The gate requirement is  $\tOrd{{\rm poly} (p,s_A,\Lambda,1/\delta)^T \log N}$ quantum gates.
In addition, if the quantum states $\ket { \vecx^{(t)}}$ at every step satisfy the error model given in Assumption \ref{assumeError} and a simulation method for the operator $\opD^{(t)}$ is provided according to Lemma \ref{lemmaSampleD},
then there exists a quantum algorithm that requires
 \begin{equation}
 \Ord{ \frac{ p^{5T} \Lambda^{2T}}{\delta^{4T}}}
\end{equation}
 copies of the initial state $\ket {\vecx^{(0)}}$ and $\tOrd{ \frac{p^{5T}\Lambda^{3T}s_A^T}{\delta^{4T}} \log N}$ quantum gates.
\end{result}
When performing multiple steps, $t=0,\dots,T$,
the step size parameter $\eta$ is usually decreased as one gets closer to the target, i.e.~$\eta \to \eta^{(t)}$.
Let $\delta>0$ be the final desired accuracy after $T$ steps.
Each step incurs the error $ \epsilon^{(t)}=\epsilon^{(t-1)} + \eta^{(t)} \epsilon_D^{(t)}$, i.e., the error from the previous step and the gradient error. Hence the accumulated error is $ \epsilon^{(t)}= \epsilon^{(0)} +  \sum_{t'=0}^{t-1} \eta^{(t')} \epsilon_D^{(t')}$. At step $T$ this error shall be $ \epsilon^{(T)} \leq \delta$. Assume for the discussion $\epsilon^{(0)} =0$. To achieve this final error $\delta$, choose the desired error of each gradient multiplication $\epsilon_D^{(t)} = \delta/T\eta^{(t)} $. Using this $\epsilon_D^{(t)}$, we have for the accumulated error $\epsilon^{(t)} = t\delta/T\leq \delta$ for $t\leq T$.
For a single gradient descent step according to Result \ref{resultMain}, the
number of copies required is then, using $\epsilon_D= \delta/T\eta$,
\begin{equation}
\Ord{{\rm poly} (p,s_A,\Lambda,1/\delta)},
 \end{equation}
where we have used the assumption $T \eta \leq T \max_t\eta^{(t)} = \Ord{1}$ and $\Lambda = \max_t  \Lambda_D^{(t)}$. Thus, $T$ iterations of the gradient descent method require,
\begin{equation}
\Ord{{\rm poly} (p,s_A,\Lambda, 1/\delta)^T}
\end{equation}
copies of the initial state $\ket{\vecx^{(0)}}$. Each copy requires 
$\tOrd{{\rm poly} (p,s_A,\Lambda_D) \log N}$ gates thus the overall gate requirement is  $\tOrd{{\rm poly} (p,s_A,\Lambda,1/\delta)^T \log N}$ quantum gates.
Now, take the Assumption \ref{assumeError} and the simulation method provided by Lemma \ref{lemmaSampleD}.
For a single gradient descent step according to Result \ref{resultMainResource}, the
number of copies required is then
\begin{equation}
 \Ord{p^5  \Lambda_D^2 \frac{(T\eta)^4}{ \delta^4}}
 = \Ord{\frac{p^5  \Lambda^2 }{\delta^4}}.
\end{equation}
For $T$ iterations of the gradient descent method, we need at most,
\begin{equation}
\Ord{ \frac{p^{5T}\Lambda^{2T}}{\delta^{4T}}},
\end{equation}
copies of the initial state $\ket{\vecx^{(0)}}$. The gate requirement is $\tOrd{ \frac{p^{5T}\Lambda^{3T}s_A^T}{\delta^{4T}} \log N}$.

Thus, for multiple steps a number of copies that is exponentiated by the number of steps is required. While this upper bound potentially can be improved significantly, it is obvious that any exponential growth of the resources in both space and run time with the number of steps is prohibitive. A possible solution is to look at extensions of the gradient descent method which exhibit faster convergence to the minimum in only a few steps. One such option is Newton's method, which requires inverting the Hessian matrix of $f(\vecx)$ at point $\vecx$, an operation that can become expensive on a classical computer while being well suited for a quantum computer.

\section{Quantum Newton's method}

A quantum algorithm for Newton's method follows the quantum gradient descent scheme, but in addition to the conditional implementation of $\opD \ket{\statex} = \ket{\nabla f(\statex)}$, one also applies
an operator $\opH^{-1}$ to $ \ket{\nabla f(\statex)}$ which represents the inverse Hessian matrix.
Our aim (classically and quantumly) here is to invert the well-conditioned subspace of $\opH$ and not take into account zero eigenvalues or eigenvalues that are extremely (exponentially) small. Thus in this work, by
$\opH^{-1}$ we denote the inverse on this well-conditioned subspace.
More formally, we define the well-conditioned subspace of $\opH$ such that $\Lambda_{H^{-1}} := \Vert \opH^{-1}\Vert$ is a constant.
Without loss of generality, we also  assume that $\Vert \opH \Vert = \Ord{1}$ \cite{Harrow09, Wiebe12}.

First, we use an assumption about the simulation of the Hessian.
\begin{assume}\label{assumeSimH}
There exists a quantum simulation method for the operator $\opH^{(t)}$ with $\Vert \opH \Vert \leq \Lambda_H$. We can simulate the controlled operator $\eu^{-\im \ket 1 \bra 1 \otimes \opH^{(t)} \tau}$ to accuracy $\epsilon_H$ using $\Ord{{\rm poly} (p,s_A,\Lambda_H,\tau, 1/\epsilon_H)}$ copies of the state $\ket{ \vecx^{(t)}}$ and queries to the oracles for $\opA$. In addition, the method requires 
$\tOrd{{\rm poly} (p,s_A,\Lambda_H) \log N}$ basic quantum gates for each copy.
\end{assume}
A Newton step can be performed, assuming this simulation method for the Hessian and the gradient simulation as before. 
\begin{result}[Single Newton step] \label{resultNewton}
 Given quantum states $\ket{ \statex^{(t)}}$ encoding the current solution at time step $t$ to Problem \ref{problem1} in amplitude encoding to accuracy $\epsilon^{(t)}>0$, as well as ancilla qubits. Let the gradient operator corresponding to the state
$\ket{ \statex^{(t)}}$ be given by $\opD^{(t)}$
and the Hessian matrix corresponding to the state
$\ket{ \statex^{(t)}}$ be given by $\opH^{(t)}$.
Assume quantum algorithms for the simulation of $\opD^{(t)}$ and $\opH^{(t)}$ according to Assumptions \ref{assumeSimD} and \ref{assumeSimH} exist. 
Let $\Vert \opD^{(t)}\Vert \leq \Lambda_D $,
$:= \Vert \opH^{(t)}\Vert \leq \Lambda_{H} $, and
$\Vert (\opH^{(t)})^{-1}\Vert \leq \Lambda_{H^{-1}}$.
Let a step size be given by $0<\eta^{(t)}<1/(2 \max\{\Lambda_D,\Lambda_{H^{-1}},\Lambda_D \Lambda_{H^{-1}} \})$.
Additionally, let  $\epsilon_{\rm nwt}>0$.
Then there exists a quantum algorithm using Oracles (2) and (3) such that a single step of the Newton's  method of step size $\eta^{(t)}$ to prepare an improved normalized quantum state $\ket{\vecx^{(t+1)}}\propto
\left( \ket{\statex^{(t)}} - \eta^{(t)} (\opH^{(t)})^{-1} \opD^{(t)} \ket{\statex^{(t)}} \right)$ to accuracy $\epsilon^{(t+1)}= \bigO(\eta^{(t)} \epsilon_{\rm nwt}^{(t)}  +\epsilon^{(t)})$  can be performed. 
This quantum algorithm requires
$\Ord{{\rm poly} (p,s_A,\Lambda,1/\epsilon_{\rm nwt})}$ 
copies of the state $\ket{ \vecx^{(t)}}$ and queries to the oracle of $\opA$, and $\tOrd{{\rm poly} (p,s_A,\Lambda,1/\epsilon_{\rm nwt}) \log N}$  basic quantum gates, 
where $\Lambda := \max\{\Lambda_D,\Lambda_H \}$.
The success probability of this algorithm is lower-bounded by $1/16$.
\end{result}

The quantum Newton step continues as follows, similar to the steps for the gradient as above.
Starting with the current solution, prepare the state
\begin{equation}
\ket{\statex^{(t)}}\left(\cos\theta \ket{0}_g  - \im \sin \theta \ket{1}_g \right ) \ket {0\dots0} \ket {0\dots0},  \label{eq:prep}
\end{equation}
where the final registers contain ancilla qubits for the gradient matrix multiplication and the Hessian matrix inversion.
With the eigenstates $\ket{u_j(\opD)}$ and the eigenvalues by $\lambda_j(\opD)$ of $\opD$,
the gradient operator phase estimation obtains:
\begin{equation}
 \cos\theta\ket{\vecx}  \ket{0}_g \ket{0\dots 0} +\im \sin \theta  \sum_j \beta_j \ket{u_j(\opD)}\ket{1}_g \ket{\tilde \lambda_j(\opD)},
\end{equation}
where $\beta_j = \langle u_j(\opD) \ket{\vecx}$. Here, $\tilde \lambda_j(\opD)$ is an approximation to $\lambda_j(\opD)$ with accuracy $\epsilon_{\rm nwt}$.
With the eigenstates of $\opH$ given by $\ket{u_j(\opH)}$ and the eigenvalues by $\lambda_j(\opH)$,
 the Hessian operator phase estimation conditioned on the ancilla being $\ket{1}_g$ obtains
\begin{equation}
\cos\theta \ket{\vecx} \ket{0}_g\ket{0\dots 0} +\im \sin \theta\sum_{j,j'}
 \beta_{j'}' \beta_j \ket{u_{j'}(\opH)}  \ket{1}_g  \ket{\tilde \lambda_j(\opD)} \ket{\tilde \lambda_{j'}(\opH)} .
\end{equation}
Here, $ \beta_{jj'}'= \bra{u_{j'}(\opH)} u_j(\opD) \rangle $ and
$\tilde \lambda_j(\opH)$ is an approximation to $\lambda_j(\opH)$ with accuracy $\epsilon_{\rm nwt}$. Perform
the conditional rotation of an ancilla (subscript $d$) for the derivative operator
\begin{eqnarray}
\ket{1}_g \ket{ \tilde\lambda_j(\opD) } \ket{0}_d \to \ket{1}_g \ket{ \tilde\lambda_j(\opD) }  \left(\sqrt{1-\left(\xi_D \tilde \lambda_j(\opD) \right)^2}\ket{0}_d+ \xi_D \tilde \lambda_j(\opD) \ket{1}_d \right),\nonumber \\
\end{eqnarray}
with $\xi_D$ such that $1/\Lambda_D > \xi_D \geq \eta$ as before. This ancilla is set to $\ket{1}_d$ for the $\ket{0}_g$ path. In addition, perform a conditional rotation of another ancilla (subscript $h$) for the eigenvalues in the well-conditioned subspace of $\opH$ \cite{Harrow09}
\begin{eqnarray}
\ket{1}_g \ket{\tilde \lambda_{j'}(\opH) } \ket{0}_h  \to  \ket{1}_g \ket{ \tilde \lambda_{j'}(\opH) }  \left(\sqrt{1-\left(\frac{\xi_H}{\tilde \lambda_{j'}(\opH)}\right)^2}\ket{0}_h + \frac{\xi_H}{ \tilde \lambda_{j'}(\opH)} \ket{1}_h\right).\nonumber \\
\end{eqnarray}
Here, $\xi_H$ is chosen such that $1/\Lambda_{H^{-1}} > \xi_H \geq \eta$, from $\frac{\xi_H}{\lambda_{j'}(\opH)} \leq 1$ for all $j'$.
The next step uncomputes the eigenvalue registers by running the phase estimations in reverse.
 A combined measurement of the ancillas in the state $\ket {\rm yes} \ket{1}_d\ket{1}_h$ arrives at
\begin{eqnarray}
& \cos\theta \ket{\vecx} - \xi_D \xi_H  \sin \theta  \sum_{jj'}  \beta_{j'}' \beta_{j}   \frac{\tilde \lambda_j(\opD) }{\tilde \lambda_{j'}(\opH)}  \ket{u_{j'}(\opH)} \\ &\quad \approx  \cos\theta \ket{\vecx} - \xi_D \xi_H  \sin \theta  \opH^{-1} \opD \ket{\vecx } .
 \end{eqnarray}
Similar to before, choosing $\theta$ such that
\begin{equation}
\cos \theta =\frac{1}{\sqrt{1+ \frac{\eta^2}{ \xi_D^2\xi_H^2}}},\
\sin \theta = \frac{\eta}{\xi_D \xi_H \sqrt{1+\frac{\eta^2}{ \xi_D^2\xi_H^2}}},
\end{equation}
results in
\begin{equation*}
\ket{\vecx^{(t+1)}}=\frac{1}{C_H^{(t+1)}}\left( \ket{\statex^{(t)}} - \eta \ket{\opH^{-1}\nabla f(\statex^{(t)})} \right),
\end{equation*}
with a normalization factor
\begin{equation}
 \left( C_H^{(t+1)}\right)^2 = 1- 2\eta \bra{ \vecx^{(t)}} \opH^{-1}\opD \ket{ \vecx^{(t)}} +\eta^2 \bra{\vecx^{(t)}} \opD \opH^{-2} \opD \ket{\vecx^{(t)}}.
 \end{equation}
This corresponds to the expressions of the gradient descent method, see Eq.~(\ref{Eq:updatedstate}), with the only difference that instead of the gradient, Newton's direction is used to update the previous state.
The probability of success of the $\ket{\rm yes} \ket 1_d \ket 1_h$ measurement is given by
\begin{equation}
P_{\mathrm{yes},1,1}^{\rm nwt} =  \frac{1}{2\left(1+\frac{\eta^2}{ \xi_D^2\xi_H^2} \right)}\left( C_H^{(t+1)}\right)^2.
\end{equation}
With the same argument as in Section \ref{Sec:singlestep} we can bound the number of repetitions needed due to the non-deterministic outcome by the eigenvalues of the operators.
One can express $\left(C_H^{(t+1)} \right)^2 = 1-2\eta \Delta_H \cos \varphi_H +\eta^2 \Delta^2_H$,
 where $\bra{\vecx^{(t)}} \opD \opH^{-2} \opD \ket{\vecx^{(t)}}:=\Delta_H$ and
$\varphi_H$ is the angle between $\ket{\vecx^{(t)}}$ and $\opH^{-1}\opD\ket{\vecx^{(t)}}$.
For a small step size, this success probability can be lower bounded.
The angle $\varphi_H = 0$ achieves the lower bound
\begin{eqnarray}
\left(C_H^{(t+1)} \right)^2 \geq \left (1-\eta \Delta_H \right)^2.
\end{eqnarray}
With $0<\eta\leq 1/(2\Lambda_D\Lambda_{H^{-1}})$ by assumption and $\Delta_H \leq \Lambda_D\Lambda_{H^{-1}}$ by definition, we have $\left(C_H^{(t+1)} \right)^2 \geq 1/4$. In addition,  $1+\eta^2/(\xi_D^2\xi_H^2) \leq 2$ because $\eta \leq \min \{\xi_D,\xi_H \}$ by the choice of $\xi_D$ and $\xi_H$. Thus
$P_{\mathrm{yes},1,1}^{\rm nwt}  >1/16$.
Thus, the upper bound for the number of repetitions needed, which is $\Ord{1/P_{\mathrm{yes},1,1}^{\rm nwt}   }$, is at most $\Ord{1}$ in this case.

As before, we can concretely specify the resources via the sample-based simulation method of the gradient operator and the Hessian. 
\begin{result}\label{resultNewtonResource}
In the setting of Result \ref{resultNewton}, include the Assumption \ref{assumeError} for the errors of the current state and a simulation method for the operator $\opD^{(t)}$ and the operator $\opH^{(t)}$ according to Lemmas 
\ref{lemmaSampleD} and \ref{lemmaSampleH}, respectively.
The quantum algorithm to perform Result \ref{resultNewton} then requires
$\Ord{\frac{p^9 \Lambda^2}{\epsilon_{\rm nwt}^4}}$
copies of $\ket{ \statex^{(t)}}$, where $\Lambda := \max\{\Lambda_D,\Lambda_H \}$.
The required number of queries to the oracle for $\opA$ is given by
$\Ord{p^8 \Lambda^3  s_A/\epsilon_{\rm nwt} ^4 }$
and the number of elementary operations is given by $\tOrd{p^{10} \Lambda^3 s_A \log N / \epsilon_{\rm nwt}^4}$.
\end{result}

The Hessian is given from Eq.~(\ref{Eq:hess}) by
\begin{equation}
\matH =  \matH_1 + \opD,
\label{Eq:fiderof}
\end{equation}
with $\matH_1$ and $\opD$ as above.
Let the norm of the Hessian $\matH$ be $\Lambda_H := \Vert \opH\Vert$, which also bounds the norm of $\matH_1$.
To obtain the eigenvalues of $\opH$ via phase estimation, we exponentiate $\opH$ via exponentiating the matrices $\matH_1$ and $\opD$ sequentially using the standard Lie product formula \cite{Childs2017}
\begin{equation}
\eu^{\im \opH \Delta t}= \eu^{\im \matH_1  \Delta t}\eu^{\im \opD \Delta t} + \Ord{ \Lambda^2 \Delta t^2}.
\end{equation}
To implement the individual exponentiations themselves we use a similar trick as before. We associate a simulatable auxiliary operator $\opM_{H_1}$  with
$\matH_1$ and, as before, the operator $\opM_D$ with $\opD$.
For the part $\matH_1$, the corresponding operator $\opM_{H_1}$
is given by
\begin{eqnarray}\label{eqMHA}
\opM_{H_1}&=& \sum_{\alpha=1}^K \sum_{ j\neq k}^p  \left(\bigotimes_{i\neq j,k}^p \mathbf{A}^{\alpha}_i \right) \otimes
 \left[ \left(\mathcal{I}\otimes \matA^{\alpha}_k \right) S \left(\mathcal{I}\otimes \matA^{\alpha}_j \right) \right].
\end{eqnarray}
Here, $S$ is the swap matrix between the last two registers.
Let $\sigma$ be an arbitrary state on which the matrix exponential $\eu^{\im \matH_1\Delta t}$ shall be applied on and $\rho=\ketbra{\vecx}{\vecx}$ be the current state. The relationship between the operator $\matH_1$ and $\opM_{H_1}$  is given by
\begin{equation}
\matH_1 \sigma =  \mathrm{tr}_{1} \{ \opM_{H_1} (\rho^{\otimes p-1} \otimes \sigma )\; \},
\end{equation}
similar to Eq.~(\ref{Eq:proxy}).
 This matrix $\opM_{H_1}$ has sparsity $p^2 s_A$ if the  matrix $\opA$ has sparsity $s_A$ and can be simulated via simulations of $\opA$ via the Oracles (\ref{oracleA}) and (\ref{oracleSparse}), see \ref{appendixSim}, Lemma \ref{Lm:AandHA}.  We can use multiple copies of the current state $\rho=\ketbra{\vecx}{\vecx}$ to perform
\begin{eqnarray}
{\rm tr}_{1\cdots p-1} \{ \eu^{- \im \opM_{H_1} \Delta t}\, (\rho \otimes \cdots \otimes \rho)  \otimes \sigma\, \eu^{ \im \opM_{H_1} \Delta t} \} \nonumber \\
 \quad = \mathcal{I} - i \Delta t[\matH_1, \sigma ] + \Ord{\Delta t^2}
\approx \eu^{-\im \matH_1 \Delta t}\, \sigma\, \eu^{\im \matH_1 \Delta t}.
\end{eqnarray}
The error for a small time step arises from the sample-based simulation method, the Lie product formula, and from errors from the current solution $\ket{\vecx}$.
As shown in \ref{App:SamBasedHam}, Lemma \ref{lemmaSampleH}, we require $\Ord{ \frac{p^9 \Lambda_H^2}{\epsilon_{\rm nwt}^4} }$ copies of the current state to perform phase estimation of $\opH$ to accuracy $\epsilon_{\rm nwt}$.

For the Newton's method the performance is determined from computing the gradient, inverting the Hessian, and subsequent vector addition as before.
If the step-size is chosen appropriately, the post-selection for the gradient, Hessian, and taking the Newton step succeeds with constant probability.
The phase estimation for the gradient to accuracy $\epsilon_{\rm nwt}$ requires $ \Ord{ \frac{p^5\Lambda_D^2 }{\epsilon_{\rm nwt}^4} }$ copies.
The phase estimation for the Hessian to accuracy $\epsilon_{\rm nwt}$ requires $ \Ord{ \frac{p^9\Lambda_H^2 }{\epsilon_{\rm nwt}^4} }$ copies.
Thus, the number of copies for a single Newton step
of step size $\eta<1/(2 \max\{\Lambda_D,\Lambda_{H^{-1}},\Lambda_D \Lambda_{H^{-1}} \})$
to prepare an improved solution is
\begin{eqnarray}
 \Ord{\frac{p^5\Lambda_D^2 }{\epsilon_{\rm nwt}^4} +  \frac{p^9\Lambda_H^2 }{\epsilon_{\rm nwt}^4}}  =  \Ord{\frac{p^9 \Lambda^2}{\epsilon_{\rm nwt}^4} }.
\end{eqnarray}
The two phase estimation requirements  contribute additively and we use $\Lambda = \max \{ \Lambda_D,\Lambda_H \}$.
The required number queries to the oracle for $\opA$ is given by
$\Ord{p^8 \Lambda^3  s_A/\epsilon_{\rm nwt} ^4 }$
and the number of elementary operations is given by $\tOrd{p^{10}\Lambda^3 s_A \log N / \epsilon_{\rm nwt}^4}$, adding the requirements for gradient and Hessian and using the maximized $\Lambda$.
The improved solution is prepared to accuracy $\bigO(\eta \epsilon_{\rm nwt}+\epsilon)$, as the current solution is given to accuracy $\epsilon$.

Similar to the discussion on multiple steps for gradient descent above, for $T$ iterations of Newton's method to final accuracy $\delta$, we need at most
\begin{equation} \label{eqNewtonMulti}
 \Ord{ \frac{p^{9T} \Lambda^{2T} }{\delta^{4T}}}
\end{equation}
copies of the initial state $\ket{\vecx^{(0)}}$, where we choose $\Lambda_D$ and $\Lambda_H$ as their respective maximum for all time steps $t$.
The gate complexity is $\tOrd{ \frac{p^{10T}\Lambda^{3T}s_A^T}{\delta^{4T}} \log N}$,  logarithmic in the dimension $N$.

This means that the required number of copies of the initial state depends exponentially on the number of steps $T$. However, recall that in Newton's method in the vicinity of an optimal solution $\vecx^\ast$ the accuracy $\Delta := \vert \vecx^\ast - \vecx^{(T)}\vert$
of the estimate often improves quadratically with the number of iterations, $\Delta \propto \bigO( \exp(-T^2))$. This convergence is for example discussed for unconstrained convex problems in \cite{boyd04}. Theorem 3.5 in Ref.~\cite{Nocedal2006} shows that if the function is twice differentiable, Lipschitz continuous in the neighborhood of a solution, and the Hessian positive definite at the solution then the convergence of Newton's method is quadratic. See also second-order sufficient conditions in Theorem 2.4 therein.
For projected methods, the convergence properties often translate under similar conditions. For example
for optimizing inequality constraints via  Newton's method one obtains a quadratic convergence \cite{Bertsekas1982}.
For proximal Newton methods, which are generalized versions of projected Newton methods, one obtains quadratic convergence under the assumption of strong convexity around the optimum and Lipshitz continuity of the Hessian \cite{Lee2012}.
Thus in some cases, even though the quantum algorithm requires a number of initial copies that grows exponentially in the number of iterations required, the accuracy of the approximate solution yielded by the algorithm
can improve even faster in the number of iterations.

\section{Inhomogeneous polynomials }
\label{sec:inhpoly}

Our methods can be extended to polynomials that are of odd order and also inhomogeneous. An example of a homogeneous polynomial of odd order is $x_1^5 x_2^2 + x_1^3 x_2^4$ and an example of an inhomogeneous polynomial is $x_1^5 x_2^2 + x_1^2 x_2^2$. A detailed discussion of such polynomials will be left for future work. We provide a discussion for a subset of polynomials we can solve in a similar fashion to the homogeneous, even polynomials in Problem \ref{problem1}. The problem can be posed as:
\begin{problem} \label{problem2} Let $f(\statex)$ be a polynomial
for which
\begin{equation}
f(\statex) = f_{\rm hom, even}(\statex) +  f_{\rm inhom}(\statex),
\end{equation}
where $f_{\rm hom,even} (\statex)$ is as in Problem \ref{problem1}.
Let the inhomogeneous part of the objective function be given by
\begin{equation}
f_{\rm inhom}(\mathbf{x}) =
\sum_{j=1}^{p-1} \left(\mathbf{c}_j^T \mathbf{x} \right) \prod\limits_{i=1}^{j-1} \left(\mathbf{x}^T {B}_{ij} \mathbf{x} \right),
\end{equation}
where the vector $\mathbf{c}_j$ and the symmetric matrices $B_{ij}$ define the polynomial. Given an initial guess $\mathbf{x}_0$, solve the problem
\[ \min_\vecx f(\vecx) , \]
subject to the constraint $\vecx^T\vecx = 1$ by finding a local minimum $\mathbf{x}^*$.
Let the assumptions of Problem \ref{problem1} also apply to the problem defined by $f(\mathbf{x})$.
\end{problem}
The term $f_{\rm inhom}(\mathbf{x})$ allows us to represent a class of monomials of uneven degree $\leq 2p-1$.
This class is essentially a sum of  homogeneous even polynomials $\prod\limits_{i=1}^{j-1} \left(\mathbf{x}^T {B}_{ij} \mathbf{x} \right)$, with the term $\mathbf{c}_j^T \mathbf{x}$ adding inhomogeneities.
Thus, with $f_{\rm hom, even}(\mathbf{x}) $ from above, we can represent combinations of homogeneous even polynomials with additional inhomogeneous terms.
In practice, the efficient sparse Hamiltonian simulation methods impose restrictions on the number of inhomogeneous terms that can be efficiently optimized. Sparsity of the matrices $B_{ij}$ implies a relatively small number of inhomogeneous monomials.

The optimization of functions containing terms of the form $f_{\rm inhom}(\mathbf{x})$ can be performed via additional simulation terms and vector additions.
In order to perform the inhomogeneous updates we require the following ingredients:
\begin{itemize}
	\item We analytically compute the respective gradient and Hessian of the function $f_{\rm inhom}(\mathbf{x})$ similar to the homogeneous part.
	\item We simulate the time evolution under the respective gradient and Hessian of the function $f_{\rm inhom}(\mathbf{x})$, using similar quantum state exponentiation methods as described in the homogeneous setting. We have to use additional subroutines to simulate terms that contain inner products of the form $\bra \vecx \mathbf c_j \rangle$, and outer products  such as $\ketbra{\vecx}{ \mathbf c_j} + \ketbra{ \mathbf c_j}{ \vecx}$.

	\item  We perform additional vector additions related to the states $\mathbf c_j$. For example the function $f_{\rm hom, even}(\vecx) + \mathbf{c}^T \vecx$ requires one additional vector addition.  Before adding the vectors we have to conditionally apply the gradient and Hessian operators to the $\mathbf x$ and $\mathbf c_j$ respectively. This requires in each iteration a number of copies of the $\mathbf c_j$ and $\mathbf x$. As we require hence in every state another copy of the states this adds another tree of states to the resource requirements, since we need to build these in parallel to the main algorithm. However, this does not change the overall scaling of the runtime of the whole algorithm.

	\item We perform similar to the homogeneous step matrix-vector multiplications and matrix inversions via a conditional rotations on the eigenvalue registers and postselection. This will result in different success probabilities for the gradient and the Hessian operator than presented above.

	\item Finally we perform a measurement in the yes/no-basis as before and perform the vector addition of the current solution and the step update.
\end{itemize}

All of these steps will add a computational overhead due to additional matrix multiplications and vector additions. This computational overhead will be discussed in a future work. Similarly to the homogeneous case the number of computational steps scales exponentially with the number of iterations $T$ and logarithmically in the dimension of the solution vector.

\section{Discussion and conclusion}

The present work has considered iterative polynomial optimization in a quantum computing framework. We have developed quantum algorithms for the optimization of a class of polynomials under spherical constraints, for which we can find expressions for the gradient and the Hessian matrix. The class of polynomials we can optimize is constrained by sparsity conditions of Hamiltonian simulation methods used here. Beyond polynomials, one can envision a setting where copies of quantum states representing the current solution are consumed for evaluating the first derivative of the objective function. If we can implement the operation
\begin{equation}
\ket{\vecx} \otimes \cdots \otimes \ket{\vecx} \ket{0} \mapsto \ket{\psi}\ket{\nabla f(\vecx)},
\end{equation}
where $\ket{\psi}$ is an arbitrary ``garbage" state, we can use the same basic gradient descent steps as discussed in this work with a similar performance as presented here for polynomial optimization.

In reference \cite{Childs2015}, Childs {\it et al.} presented an exponential improvement of the matrix inversion error dependence, $1/\epsilon \to \log 1/\epsilon$ by using approximation polynomials instead of phase estimation. 
Further exponentially-precise Hamiltonian simulation methods were shown in \cite{Low2016qubitization,Low2017spectral}, 
which may find application in the gradient descent problem. In addition, variable-time amplitude amplification can lead to further quadratic speedups \cite{Brassard2002,Ambainis2010,Childs2015}. The connection of the present work to the semi-definite programming quantum solver in \cite{Brandao2017,Apeldoorn2017} is another potential avenue of study.

Our optimization algorithms scale exponentially in the number of steps performed.
Such a performance is in general expected since the problem we attempt to solve is QMA-hard, since we can reduce certain
$k$-local Hamiltonian decision problem to it, see e.g.~\cite{bookatz2012qma} for a definition.
However, we envision
a few scenarios where the algorithm can nevertheless be useful. One scenario is when the
size of the vectors dominates the dependence on other parameters such as condition
number and error raised to the power of the number of steps.  In this case, classical
computation is prohibitively expensive, while our quantum method scales logarithmically
in the size of the vectors and could allow optimization.  Often, a constant number of steps
can yield a significant improvement on the initial guess, even without finding a local
minimum, whereas the problem would be intractable on a classical computer.
Another case where our quantum algorithm yields potential speed ups is the
application of Newton's method to (locally) convex problems. In such problems,
the number of iterations of Newton's method
required to find a highly accurate solution is only weakly dependent on the dimension
of the system \cite{boyd04}, and often the minimum can be found in around 5-20 steps.
Yet the standard Newton's method is also well known to fail in high-dimensional spaces due to the prevalence of saddle-points, see for example~\cite{martens2010deep}. While this issue is classically not trivial to solve, our quantum method allows for an easy extension to the saddle-free Newton's method~\cite{dauphin2014identifying}. This is done by simply replacing the inverse of the eigenvalues with the inverse of the absolute values of the eigenvalues. Thereby the sign of the eigenvalues around a saddle-point is not changed and hence the algorithm takes steps in the correct direction. Therefore, our quantum algorithm is applicable to a wider range of problems with only slight adaptions.

In summary, the optimization algorithms presented here yield a performance $\bigO({\rm polylog}(N))$
in the dimension $N$ of the vector space over which the optimization is performed,
as long as the number of iterations required is small. When searching for a small number
of solutions in a featureless landscape, a large number iterations is required, and the
algorithms scale polynomially in $N$: in particular, the algorithms cannot do better than
Grover search, which finds $M$ solutions in an essentially `gradient-free' landscape in
time $O(\sqrt{N/M})$.
Further research may find a possible application of quantum gradient algorithms in machine learning for the training of deep neural networks \cite{Hinton06, Bengio09},  which exhibit a large number of parameters.
Gradient descent methods indeed lead to good solutions after a few iterations,
and quantum gradient descent algorithms may provide exponential improvements over their classical counterparts.

\textit{Acknowledgement}$-$
We thank George Siopsis and Simon Benjamin for valuable discussions. We also thank two anonymous referees for their detailed reading of the manuscript and valuable comments. PR and LW thank Joseph Fitzsimons for a hospitable stay at the Centre for Quantum Technologies, Singapore. MS and FP acknowledge support by the South African Research Chair Initiative of the Department of Science and Technology and National Research Foundation.
\\

\appendix

\section{Simulating the gradient matrix $\opM_D$ and the Hessian part $\opM_{H_1}$ }
\label{appendixSim}

We state the complexities as far as required for the discussion in the main part and note that some dependencies may be polynomially improved. We first summarize the known Hamiltonian simulation result. 
\begin{theorem}[Sparse Hamiltonian simulation \cite{Berry2012}]\label{thmHS}
Given the Oracles (\ref{oracleA}) and (\ref{oracleSparse}) for the $s_A$-sparse matrix $\opA$ in the vector space of dimension $N^p$, $s_A\geq 1$. Let $\Vert \opA \Vert \leq \Lambda_A$, $\Lambda_A\geq 1$, and $\Vert \opA\Vert_{\max} := \max_{ij} \vert \opA_{ij} \vert $. There exists a quantum algorithm to simulate $\eu^{-\im \opA \tau}$ to accuracy $0<\epsilon<1$ using $\Ord{c_A \tau}$ queries to the oracles and $\tOrd{g_A \tau p \log N}$ single- and two-qubit quantum gates, where $c_A :=\frac{\Lambda_A}{\sqrt{\epsilon}} + \Vert \opA\Vert_{\max} s_A +1$. Since $\Vert \opA\Vert_{\max} \leq \Lambda_A$, we can relax the upper bound as $c_A =\Ord{ \frac{\Lambda_A s_A}{\sqrt \epsilon}}$, which is the bound used in this work. 
\end{theorem}

We use this result to simulate the matrices that occur in the problem of enacting the gradient and Hessian operators. 

\begin{lemma}\label{Lm:AandMD}
Let the derivative auxiliary operator $\opMD$ be given by Eq.~(\ref{eqMD}), with the vector space dimension $N^p$. Let $ \Vert \opA \Vert \leq \Lambda_A$ and $ \Lambda_A \geq 1$. There exists a
 quantum computation using Oracles (\ref{oracleA}) and (\ref{oracleSparse}) such that
$\eu^{-\im \opMD \tau}$ can be simulated. For simulating a time $\tau$ to accuracy $\epsilon$, we require $\Ord{p^3 \Lambda_A s_A \tau^2/\epsilon}$ queries to the oracles for $\opA$ and $\tOrd{p^4 \Lambda_A s_A \tau^2 \log N/\epsilon}$ quantum gates.
\end{lemma}
\begin{proof}
 From Theorem \ref{thmHS}, $\eu^{-\im \matA \tau}$ can be simulated to error $\epsilon$ in $\Ord{{\Lambda_A s_A \tau}/{\sqrt \epsilon}}$ queries to the oracles for $\opA$ and using $\tOrd{\Lambda_A s_A  \tau\ p \log(N)/\sqrt \epsilon}$ gates. The task is to efficiently simulate $\eu^{-\im \opMD \tau}$.
The matrix $\matA$ can be expressed as:
\begin{equation}
\matA = \sum_{\alpha=1}^K \matA^{\alpha}_1 \otimes \cdots \otimes \matA^{\alpha}_p.
\end{equation}
Note that this tensor decomposition is formal but not explicitly required for simulating the gradient matrix $\opMD$.
The gradient matrix is given by Eq.~(\ref{eqMD}), where we can simply interchange the summations to obtain:
\begin{equation}
 \opMD =  \sum\limits^p_{j=1}  \sum\limits_{\alpha=1}^K\left( \bigotimes^p_{{i=1 \atop i\neq j}} \opA_{i}^{\alpha} \right )\otimes \opA_{j}^{\alpha} =:  \sum\limits^p_{j=1}  M_j.
\end{equation}
For the discussion, we assume that $\bigotimes^p_{{i=1 \atop i\neq j}} \opA_{i}^{\alpha}$ is ordered such that the $p$-th matrix is swapped with the $j$-th matrix, which results in the same operator $\opD$, since the ordering is not important (after taking the partial trace).
Note that there exist (unitary) permutation matrices $Q_j$, such that
with $M_j =  Q_j \matA Q_j^\dagger$ we have
 \begin{equation}\label{eqPermSum}
\opMD = \sum_{j=1}^p Q_j \matA Q_j^\dagger.
 \end{equation}
Thus, the $M_j$ matrix elements can be obtained from $\matA$ matrix elements.
The permutation matrices are specified via
\begin{eqnarray} \label{eqPermutation}
\bra{i_1\cdots i_j \cdots i_p } M_j \ket{i_1'\cdots i_j' \cdots i_p'} =\\
\quad =\sum_{\alpha=1}^K (\matA^{\alpha}_1)_{i_1 i_1'} \cdots  (\matA^{\alpha}_p)_{i_j i_j'}  \cdots (\matA^{\alpha}_j)_{i_p i_p'}  \\
\quad =  \sum_{\alpha=1}^K (\matA^{\alpha}_1)_{i_1 i_1'} \cdots  (\matA^{\alpha}_j)_{i_p i_p'}  \cdots (\matA^{\alpha}_p)_{i_j i_j'}, \\
\quad = \bra{i_1\cdots i_p \cdots i_j } \matA \ket{i_1'\cdots i_p' \cdots i_j'},
\end{eqnarray}
to be
\begin{equation}
Q_j = \sum_{i_1,\dots,i_p=1}^N \ket{i_1\cdots i_j \cdots i_p } \bra{i_1\cdots i_p \cdots i_j }.
\end{equation}
To simulate a small step $\eu^{-\im \opMD \Delta t}$, use the Lie product formula for the sum in Eq.~(\ref{eqPermSum}),
 \begin{eqnarray}\label{eqMDSim}
\left \Vert \eu^{-\im \opMD \Delta t}  - \prod_{j=1}^p Q_j \eu^{-\im \matA \Delta t} Q_j^\dagger  \right \Vert =
\Ord{p^2 \Delta t^2}.
\end{eqnarray}
where the sum of $p$ terms leads to an error scaling as $p^2$, as discussed for example in \cite{Childs2017}.
Hence $\opMD$ can be simulated for $\Delta t$ via a sequence of simulations of $\matA$ sandwiched between the permutation matrices.  Each of the $p$ permutation matrices  $Q_j$ can be simulated by swapping the two $\log N$-qubit registers corresponding to the respective permutation, i.e.~involving $\log N$ swap operations. By Eq.~(\ref{eqMDSim}), a single step thus requires
$2 p \log N$ swap operations. Also, we require $p$ simulations of $\opA$ of time $\Delta t$.
The error of these simulations shall be $p \Delta t^2$ to accumulate to a total error $\Ord{p^2 \Delta t^2}$. Thus, for each simulation of $\opA$ we require $\Ord{{\Lambda_A s_A \Delta t}/{\sqrt{p \Delta t^2}}} \leq \Ord{\Lambda_A s_A}$ queries to the oracles for $\opA$ and using $\tOrd{\Lambda_A s_A  \Delta t\ p \log(N)/\sqrt{p \Delta t^2}} \leq \tOrd{\Lambda_A s_A p \log N}$ gates from Theorem \ref{thmHS}.

For multiple steps, 
 \begin{eqnarray}\label{eqMDSim}
\epsilon :=  \left \Vert \eu^{-\im \opMD \tau}  - \left( \prod_{j=1}^p Q_j \eu^{-\im \matA \frac{t}{m}} Q_j^\dagger \right )^m  \right \Vert,
\end{eqnarray}
using $\tau= m \Delta t $, we take $m=\Ord{p^2\tau^2/\epsilon}$ to achieve final accuracy $\epsilon$. This requires
in total $\Ord{m p \Lambda_A s_A}$ queries to the oracles for $\opA$ and $\tOrd{m p^2( \Lambda_A +1)s_A \log N} =\tOrd{m p^2 \Lambda_A s_A \log N}$ quantum gates, if $\Lambda_A \geq 1$.
\end{proof}
Recall that the simulation of $\opMD$ is a means to simulate the gradient operator $\opD$. The additional error incurred for $\opMD$ will be included in the error incurred by the quantum state exponentiation for simulating $\opD$, which will be discussed in \ref{App:SamBasedHam}. In addition, we note that improved methods for simulating $\opMD$ may be found via efficiently computing the sparsity function $g_{\opMD}(k,l)$ from the sparsity function of $\matA$.

We can provide a similar result for the operator appearing in Newton's method.
\begin{lemma}\label{Lm:AandHA}
Let the Hessian auxiliary operator part $\opM_{H_1}$ be given by Eq.~(\ref{eqMHA}), with the vector space dimension $N^p$. Let $ \Vert \opA \Vert \leq \Lambda_A$ and $ \Lambda_A \geq 1$. There exists an
efficient quantum computation using Oracles (\ref{oracleA}) and (\ref{oracleSparse}) such that
$\eu^{-\im \opM_{H_1} t}$ can be simulated.
For simulating a time $\tau$ to accuracy $\epsilon$, we require $\Ord{p^6 \Lambda_A s_A  \tau^2/\epsilon}$ queries to the oracles for $\opA$ and  $\tOrd{p^6 \Lambda_A  s_A \tau^2 \log N/\epsilon}$ quantum gates.
\end{lemma}
\begin{proof}
Note that the operator can be written as
\begin{equation}
\opM_{H_1} =  \sum_{ j\neq k}^p H_{jk},
\end{equation}
with
\begin{eqnarray}\label{eqHjk}
H_{jk}= \sum_{\alpha=1}^K  \left(\bigotimes_{i\neq j,k}^p \mathbf{A}^{\alpha}_i \right) \otimes
 \left[ \left(\mathcal{I}\otimes \matA^{\alpha}_k \right) S \left(\mathcal{I}\otimes \matA^{\alpha}_j \right) \right].
\end{eqnarray}
Note further that we can relate each $H_{jk}$ to the matrix $\opA$, (for $j<k$)
\begin{eqnarray} \label{eqPermutation2}
 \bra{i_1\cdots i_j\cdots i_k \cdots i_{p-1} i_p } H_{jk} \ket{i_1'\cdots i_j' \cdots i_k' \cdots i_{p-1}' i_p'} =\\
 \quad =\sum_{\alpha=1}^K (\matA^{\alpha}_1)_{i_1 i_1'} \cdots  (\matA^{\alpha}_{p})_{i_j i_j'}
  \cdots  (\matA^{\alpha}_{p-1})_{i_k i_k'}  \cdots (\matA^{\alpha}_k)_{i_p i_{p-1}'} (\matA^{\alpha}_j)_{i_{p-1} i_p'} \\
  \quad =
\bra{i_1\cdots i_{p-1}\cdots i_{p} \cdots i_{k} i_j } \opA \ket{i_1'\cdots i_p \cdots i_{p-1}' \cdots i_{k}' i_j'}   .
\end{eqnarray}
The permutation between the $p-1$ and $p$ part is owed to the additional swap matrix in Eq.~(\ref{eqHjk}). Each required permutation to relate $\opM_{H_1}$ to $\opA$  can be performed with $2 \log N$ swap operators. We use the Lie product formula as in Lemma \ref{Lm:AandMD},
 to simulate $\eu^{-\im \opM_{H_1} \Delta t}$ efficiently with error $\Ord{p^4 \Delta t^2}$ using $p^2$ small-time simulations of $\opA$.
 The error of the simulations of $\opA$ shall be $p^2 \Delta t^2$ to accumulate to a total error $\Ord{p^4 \Delta t^2}$. Thus, for each simulation of $\opA$, we require $\Ord{{\Lambda_A s_A \Delta t}/{\sqrt{p^2 \Delta t^2}}} \leq \Ord{\Lambda_A s_A}$ queries to the oracles for $\opA$ and using $\tOrd{\Lambda_A s_A  \Delta t\ p \log(N)/\sqrt{p^2 \Delta t^2}} \leq \tOrd{\Lambda_A s_A \log(N)}$ gates \cite{Berry2012}.
 We also have $2 p^2 \log N$ swap operations for a small time step.

 For multiple steps $m=p^4 \tau^2/\epsilon$ to given $\epsilon$, we thus require
 $\Ord{m p^2 \Lambda_A s_A}$ queries to the oracles of $\opA$ and $\tOrd{m (\Lambda_A +1)s_A\ p^2 \log N} \leq \tOrd{m \Lambda_As_A\ p^2 \log N}$ quantum gates, if $\Lambda_A \geq 1$.
\end{proof}
\section{Sample-based Hamiltonian simulation with erroneous Hamiltonian}\label{App:SamBasedHam}

In this work, we perform Hamiltonian simulation with a Hamiltonian that is defined by quantum states that exhibit an error. Concretely, we want to simulate the Hamiltonian
$\opD = \mathrm{tr}_{1\dots p-1} \left \{\left (\ketbra{ \vecx}{ \vecx}^{\otimes (p-1)}\otimes \mathcal{I} \right )\;  \opMD \right \}  $ using a Trotter decomposition into small steps. But instead of the exact steps, we simulate small steps with an error, i.e.~we use the Hamiltonian $ \opD_l= \mathrm{tr}_{1\dots p-1} \left \{ \left(\bigotimes_{k=1}^{p-1}\ketbra{\tilde \vecx_{lk}}{\tilde \vecx_{lk}}\otimes \mathcal{I}\right) \;  \opMD \right \} $. Here, $ \ket{\tilde \vecx_{lk}} $ are the erroneous samples.
We assume in this section that the errors are random vectors.
We also assume here that the errors are small, i.e.~the norm of the vectors is $\ll 1$ and consider contributions to leading order in the norm.
For the erroneous Hamiltonian simulation we obtain the following result.
\begin{lemma}[Erroneous sample-based Hamiltonian simulation of $\opD$]\label{lemmaSampleD}
  Given the desired Hamiltonian $\opD$ (with $\Vert \opD \Vert \leq \Lambda_D=\Ord{p \Lambda_A}$) and the actual Hamiltonians  $\opD_l$ arising from erroneous samples
  $\ket{\tilde \vecx_{lk}} $, $k=1,\dots, p$. Here, $l=1,\cdots,m$, where $m$ is the number of time steps. Assume that the samples are given with independent, bounded, random errors, i.e.~$  \ket{\tilde \vecx_{lk}} - \ket{ \vecx}= \vec v_{lk}$
  such that $\left \vert  \vec v_{lk}\right \vert_2  \ll 1$, $l = 1,\dots,m$, $k = 1,\dots,p$, as in Assumption \ref{assumeError}, where the norm is the standard Euclidian norm.
  Using the sample-based Hamiltonian simulation scheme of~\cite{Lloyd14} and the matrix simulation of $\opM_D$ described in Lemma \ref{Lm:AandMD}, we can perform Hamiltonian simulation of $\opD$. The simulation of $\eu^{-\im \opD \tau}$ for a time $\tau$ and desired error $\epsilon$ can be performed
with $m = \Ord{ p^4 \Lambda_D^2  \frac{\tau^2}{\epsilon^{2}}  } $ time steps. The number of samples needed at each time step is $p$.
The algorithm uses $\Ord{p^4 \Lambda_D^3 \tau^2  s_A/\epsilon^{2} }$ queries to the oracle for $\opA$ and $\tOrd{p^5 \Lambda_D^3 \tau^2 s_A \log N / \epsilon^{2} }$ quantum gates.
\end{lemma}
We immediately see that for $\tau=1/\epsilon$ the required number of steps is hence given by
$\Ord{ \frac{p^4 \Lambda_D^2 }{\epsilon^4}}$. Then the total number of samples required
is $\Ord{ \frac{p^5\Lambda_D^2}{\epsilon^4}}$.
\begin{proof}
Define $\beta$ to be the maximum of the norms of the sample errors, $ \beta := \max_{l,k} \vert  \vec v_{lk}\vert  \ll 1$. As stated,  $\ket{\tilde \vecx_{lk}} = \ket{ \vecx} +\vec v_{lk}$, and
\begin{equation}\label{eqBraKetExpansion}
\ketbra{\tilde \vecx_{lk}}{\tilde \vecx_{lk}} = \ketbra{ \vecx} {\vecx} + \ket{ \vecx} \vec v_{lk}^T+  \vec v_{lk}\bra{\vecx} + \vec v_{lk} \vec v_{lk}^T.
\end{equation}
Note that $\left \vert \vec v_{lk} \vec v_{lk}^T \right \vert =\Ord{\beta^2}$. We first perform an expansion of the erroneous Hamiltonian using Eq.~(\ref{eqBraKetExpansion}) and collecting some of the $\Ord{\beta^2}$ terms
\begin{eqnarray}
\opD_l &=& \mathrm{tr}_{1\dots p-1}  \left \{ \left( \otimes_{k=1}^{p-1} \ketbra{\tilde{\vecx}_{lk}} {\tilde{\vecx}_{lk}} \otimes \mathcal I \right) \opMD \right \} \\
&=&
\mathrm{tr}_{1\dots p-1}  \left \{ \sum_{i=1}^{p-1} X_{li}\; \opMD \right \} + \Ord{p^2 \beta^2}, \nonumber
 \end{eqnarray}
 where
 $X_{li} := \ketbra{\vecx}{\vecx}^{\otimes i-1} \otimes \ketbra{\tilde \vecx_{li}} {\tilde \vecx_{li}} \otimes \ketbra{\vecx}{\vecx}^{\otimes (p-1)-i}\otimes \mathcal I$.
This can be further expanded to
\begin{equation}
  \label{eq:expansion_D}
 \opD_l = \opD +  \mathrm{tr}_{1\dots p-1} \left \{ \sum_{i=1}^{p-1} V_{li}\; \opMD \right \} + \Ord{p^2 \beta^2},
 \end{equation}
where  $V_{li} := \ketbra{\vecx}{\vecx}^{\otimes i-1} \otimes \left (\ket{\vecx} \vec v_{li}^T + \vec v_{li} \bra{\vecx} \right) \otimes \ketbra{\vecx}{\vecx}^{\otimes (p-1)-i}\otimes \mathcal I$.
 To order $\Ord{\beta}$ only one state is erroneous and we have $\Ord{p}$ terms with the error at a different location in the $p-1$ state register.
 Here, we can analyze the term where only the first register is erroneous and assume that the other terms are of the same size. That is we use
 \begin{eqnarray}
 \mathrm{tr}_{1\dots p-1} \left \{ \sum_{i=1}^{p-1} V_{li} \opMD \right \} = \Ord{ p} \mathrm{tr}_{1\dots p-1}  \left \{ \left( \ket{ \vecx}\vec v_{l1}^T  \otimes\left(\ketbra{ \vecx}{ \vecx}\right )^{\otimes (p-2)} \otimes \mathcal I \right)\;  \opMD \right \}.\nonumber \\
 \end{eqnarray}
 Continue as
 \begin{eqnarray}
 \opD_l &=& \opD + \Ord{ p} \mathrm{tr}_{1\dots p-1}  \left\{ \left ( \ket{ \vecx}\vec v_{l1}^T  \otimes\left(\ketbra{ \vecx}{ \vecx}\right )^{\otimes (p-2)} \otimes \mathcal I \right) \;  \opMD \right \} +\Ord{\beta^2 p^2} \nonumber \\ &=:& \opD + \tilde \opD_l.
\end{eqnarray}
The zeroth order reproduces $\opD$ while $\tilde \opD_l$ is the erroneous part. Now we analyze the error term $\mathrm{tr}_{1\dots p-1} \left \{\left( \ket{ \vecx}\vec v_{l1}^T  \otimes\left(\ketbra{ \vecx}{ \vecx}\right )^{\otimes (p-2)}  \otimes \mathcal I \right) \;  \opMD \right \} $ as a representative of the other $\Ord{p}$ terms and using the actual matrix $\opM_D$. The general bound follows then by taking the maximum over the matrices $\matA_j^\alpha$. Continue as
\begin{eqnarray} 
\mathrm{tr}_{1\dots p-1} \left \{ \left ( \ket{ \vecx}\vec v_{l1}^T \otimes\left(\ketbra{ \vecx}{ \vecx}\right )^{\otimes (p-2)} \otimes \mathcal I \right ) \;    \sum\limits_{\alpha=1}^K \sum\limits^p_{j=1}\left( \bigotimes^p_{{i=1 \atop i\neq j}} \opA_{i}^{\alpha} \right )\otimes \opA_{j}^{\alpha}  \right \} =\nonumber \\ \quad \quad
= \Ord{p} \sum_{\alpha = 1}^K \mathrm{tr}\left \{\ket \vecx \vec v^T_{l1} \matA_1^\alpha \right \}   \prod_{i=2}^{p-1} \mathrm{tr}\left \{ \ketbra{ \vecx}{ \vecx} \matA_i^\alpha \right \}  \matA_p^{\alpha} \\
\quad \quad \equiv \Ord{p} \sum_{\alpha = 1}^K  Y_l^\alpha  \Vert \matA_1^\alpha \Vert  \prod_{i=2}^{p-1} \mathrm{tr}\left \{ \ketbra{ \vecx}{ \vecx} \matA_i^\alpha \right \}  \matA_p^{\alpha}.\label{eqErrorExplicit}
\end{eqnarray}
Here, we have another $\Ord{p}$ terms that are similar to the one explicitly written out.
Because $\vec v_{l1}$ is a random vector, we have random variables $Y_l^\alpha$ given by
\begin{equation}
\label{eqErrorRandomVariable}
Y_l^\alpha := \frac{\mathrm{tr}\{\ket \vecx \vec v_{l1}^T \matA_1^\alpha \} }{\Vert \matA_1^\alpha \Vert}.
\end{equation}
These random variables have support on $[-\infty, +\infty]$, a symmetric distribution and standard deviation $\Ord{\beta}$.
To get the final error we will use a central limit argument.

As we care about the quantum simulation of the Hamiltonian $\opD$, we need to evaluate the difference in desired and actual time evolution. Note Equation (\ref{eqPCAErrorState}) for the error of a single time step $\Delta t = t/m$:
\begin{equation}
\mathcal E' = \mathcal E_{\opMD} + \mathcal E_{\rm samp } + \mathcal E_{\ket \vecx }.
\end{equation}
As before, we can choose $\Vert \mathcal E_{\opMD} \Vert = \Ord{p^2 \Delta t^2}$, using $\Ord{p \Lambda_A  s_A }=\Ord{\Lambda_D  s_A}$ queries to the oracle for $\opA$ and using $\tOrd{p^2 \Lambda_A  s_A \log N}=\tOrd{p \Lambda_D  s_A \log N}$ gates,
see \ref{appendixSim}, Lemma \ref{Lm:AandMD}.
Regarding the error $\mathcal E_{\rm samp }$, its size is given by $\Vert\mathcal E_{\rm samp } \Vert = \Ord{\Lambda_D^2 \Delta t^2 }$ similar to \cite{Lloyd14}. The error for a small time step $\Delta t$ from these two sources can be upper bounded by $\Vert\mathcal E \Vert \leq \Vert\mathcal E_{\rm samp } \Vert + \Vert \mathcal E_{\opMD} \Vert = \Ord{p^2 \Lambda_D^2 \Delta t^2 }$. For $m$ steps this error is $m \Vert\mathcal E \Vert = \Ord{p^2 \Lambda_D^2 \tau^2/m }$

As in \cite{Childs2017}, define the remainder of an analytic function $f = \sum_{j=0}^\infty a_j x^j$ to be $\mathcal R_k(f) := \sum_{j=k+1}^\infty a_j x^j$, for $k \in \mathbbm N$. Then the error between desired and actual time evolution $\epsilon$ is
\begin{eqnarray}
\epsilon &\leq& \left \Vert \eu^{- \im\opD \tau}  - \prod_{l=1}^m \eu^{- \im \opD_l \tau/m} \right \Vert + m\Vert \mathcal E \Vert\\
&=&  \left \Vert \eu^{- \im\opD \tau}  - \prod_{l=1}^m \eu^{- \im (\opD + \tilde \opD_l) \tau/m} \right \Vert + m\Vert \mathcal E \Vert \\
&=&\left \Vert \eu^{- \im\opD \tau}  - \prod_{l=1}^m \eu^{- \im \opD \tau/m}\eu^{-\im \tilde \opD_l \tau/m} + \prod_{l=1}^m \eu^{- \im \opD \tau/m}\eu^{-\im \tilde \opD_l \tau/m}  - \prod_{l=1}^m \eu^{- \im (\opD + \tilde \opD_l) \tau/m} \right \Vert \nonumber \\
&& + m\Vert \mathcal E \Vert\\
&\leq&\left \Vert \eu^{- \im\opD \tau}  - \eu^{- \im \opD \tau} \prod_{l=1}^m \eu^{-\im \tilde \opD_l \tau/m} \right \Vert + \left \Vert  \prod_{l=1}^m  \eu^{- \im \opD \tau/m}\eu^{-\im \tilde \opD_l \tau/m}  - \prod_{l=1}^m  \eu^{- \im (\opD + \tilde \opD_l) \tau/m} \right \Vert \nonumber \\
&& + m\Vert \mathcal E \Vert \\
&\leq&\left \Vert \eu^{- \im\opD \tau}  - \eu^{- \im \opD \tau} \prod_{l=1}^m \eu^{-\im \tilde \opD_l \tau/m} \right \Vert + \Ord{\frac{\Lambda_D^2 \tau^2}{m}} + m\Vert \mathcal E \Vert \nonumber \\
&\leq&\left \Vert\mathcal I - \prod_{l=1}^m \left( \mathcal I - \im \frac{\tilde \opD_l \tau}{m} + \mathcal R_1\left( \eu^{-\im \tilde \opD_l \tau/m} \right) \right) \right \Vert + \Ord{\frac{p^2 \Lambda_D^2 \tau^2}{m}} \nonumber \\
&=&\left \Vert \sum_{l=1}^m \frac{\tilde \opD_l \tau}{m} \right \Vert +\Ord{\beta^2} + \Ord{\frac{p^2 \Lambda_D^2 \tau^2}{m}} \nonumber \\
&=& \Ord{ p^2}  \left \Vert \sum_{l=1}^m \sum_{\alpha = 1}^K  Y_l^\alpha  \Vert \matA_1^\alpha \Vert  \prod_{i=2}^{p-1} \mathrm{tr}\left \{ \ketbra{ \vecx}{ \vecx} \matA_i^\alpha \right \}  \matA_p^{\alpha} \right \Vert   \frac{\tau}{m}  \\ &&+  \Ord{p^2\beta^2} + \Ord{\frac{p^2 \Lambda_D^2 \tau^2 }{m}}. \nonumber
\end{eqnarray}
The last identity follows from Eq.~(\ref{eqErrorExplicit}), using the random variables defined in Eq.~(\ref{eqErrorRandomVariable}). Define 
 the random variables $Y_l:= K \Lambda_D \sum_{\alpha=1}^K  Y_l^\alpha $, which 
have support on $[-\infty, +\infty]$, a symmetric distribution and standard deviation $\Ord{K^2 \Lambda_D \beta}$.
 The leading term can be bounded as
\begin{eqnarray}
\left \Vert \sum_{l=1}^m  \sum_{\alpha = 1}^K  Y_l^\alpha  \Vert \matA_1^\alpha \Vert  \prod_{i=2}^{p-1} \mathrm{tr}\left \{ \ketbra{ \vecx}{ \vecx} \matA_i^\alpha \right \}  \matA_p^{\alpha} \right \Vert  
 &\leq& \left \vert \sum_{l=1}^m  Y_l \right \vert 
 \\
 &\leq& \Ord{K^2 \Lambda_D \beta  \sqrt{m}}.
\end{eqnarray}
Here, the central limit theorem allows to bound the size of the sum of the random variables $Y_l$ by $\Ord{\sqrt{m}}$ and the standard deviation $\Ord{K^2 \Lambda_D \beta}$. To simplify, we assume $K=\Ord{1}$. Putting all together, neglecting higher-order $\beta$ terms, we obtain the simulation error:
$\epsilon = \Ord{  p^2 \beta \Lambda_D  \frac{\tau}{\sqrt{m}} + p^2 \Lambda_D^2 \frac{\tau^2}{m} }$.
 Solving this for the number of samples
 $m$ gives
 \begin{equation}
m = \Ord{  \frac{ p^4 \Lambda_D^2  \tau^2}{\epsilon^2 } \chi}.
 \end{equation}
 with $\chi := (\beta^2 + 2 \epsilon/p^2 + \beta \sqrt{\beta^2+4\epsilon/p^2})/2\leq (1 + 2 + \sqrt{1+4})/2 < 3$.
 \end{proof}
The resulting dependencies can also be interpreted in comparison with black-box Hamiltonian simulation.  For black-box Hamiltonian simulation, the error in the simulated time-evolution grows only linearly in
time for an erroneous Hamiltonian~\cite{cubitt2017universal}, and hence one obtains a scaling of $\Ord{\tau/ \epsilon}$, while for the sample-based Hamiltonian scheme we have in general a dependency which is quadratically larger compared to the
black-box Hamiltonian simulation~\cite{kimmel2017hamiltonian}, which is in accordance with the error bounds we obtain here.

\begin{lemma}[Erroneous sample-based Hamiltonian simulation of $\matH_1$]\label{lemmaSampleH}
  Given the desired Hamiltonian $\matH_1$, with $\Vert \matH_1\Vert \leq \Vert \matH \Vert \leq \Lambda_H$, and the actual Hamiltonians  $(\matH_1)_l$ arising from erroneous samples.
  With the same setting as Lemma \ref{lemmaSampleD}, the simulation of $\eu^{-\im \matH_1 \tau}$ for a time $\tau$ and desired error $\epsilon$ can be performed, 
with $m = \Ord{ p^8 \Lambda_H^2  \frac{\tau^2}{\epsilon^{2}}  }$ time steps. The number of samples needed at each time step is $p$.
The algorithm uses $\Ord{p^8 \Lambda_H^3 \tau^2  s_A/\epsilon^{2} }$ queries to the oracle for $\opA$ and $\tOrd{p^{10} \Lambda_H^3 \tau^2 s_A \log N / \epsilon^{2} }$ quantum gates.
\end{lemma}
\begin{proof}
The proof is analogous to Lemma \ref{lemmaSampleD} but the error comes in $\Ord{p^2}$ contributions instead of  $\Ord{p}$ contributions.
For simplicity, we take the $p$ dependence to be $p^8$ and $p^{10}$, squaring the $p$ dependence of Lemma \ref{lemmaSampleD}, and note that it may be significantly improved. 
\end{proof}


\begin{thebibliography}{10}

\bibitem{Sra12}
Suvrit Sra, Sebastian Nowozin, and Stephen~J Wright.
\newblock {\em Optimization for machine learning}.
\newblock MIT Press, 2012.

\bibitem{martens2010deep}
James Martens.
\newblock Deep learning via hessian-free optimization.
\newblock In {\em Proceedings of the 27th International Conference on Machine
  Learning (ICML-10)}, pages 735--742, 2010.

\bibitem{boyd04}
Stephen Boyd and Lieven Vandenberghe.
\newblock {\em Convex optimization}.
\newblock Cambridge university press, 2004.

\bibitem{Nocedal2006}
J.~Nocedal and S.~Wright.
\newblock {\em Numerical Optimization}.
\newblock Springer, 2006.

\bibitem{Durr1996}
C.~D\"urr and P.~Hoyer.
\newblock A quantum algorithm for finding the minimum.
\newblock {\em arXiv preprint quant-ph/9607014}, 1996.

\bibitem{Farhi2016}
Edward Farhi and Aram~W Harrow.
\newblock Quantum supremacy through the quantum approximate optimization
  algorithm.
\newblock {\em arXiv:1602.07674}, 2016.

\bibitem{Farhi00}
Edward Farhi, Jeffrey Goldstone, Sam Gutmann, and Michael Sipser.
\newblock Quantum computation by adiabatic evolution.
\newblock {\em arXiv preprint quant-ph/0001106}, 2000.
\newblock MIT-CTP-2936.

\bibitem{Denchev12}
Vasil Denchev, Nan Ding, Hartmut Neven, and Svn Vishwanathan.
\newblock Robust classification with adiabatic quantum optimization.
\newblock In {\em Proceedings of the 29th International Conference on Machine
  Learning (ICML-12)}, pages 863--870, 2012.

\bibitem{Neven09}
Hartmut Neven, Vasil~S Denchev, Geordie Rose, and William~G Macready.
\newblock Training a large scale classifier with the quantum adiabatic
  algorithm.
\newblock {\em arXiv preprint arXiv:0912.0779}, 2009.

\bibitem{Benedetti15}
Marcello Benedetti, John Realpe-G{\'o}mez, Rupak Biswas, and Alejandro
  Perdomo-Ortiz.
\newblock Estimation of effective temperatures in a quantum annealer and its
  impact in sampling applications: A case study towards deep learning
  applications.
\newblock {\em arXiv preprint arXiv:1510.07611}, 2015.

\bibitem{Wiebe12}
Nathan Wiebe, Daniel Braun, and Seth Lloyd.
\newblock Quantum algorithm for data fitting.
\newblock {\em Physical Review Letters}, 109(5):050505, 2012.

\bibitem{Schuld16}
Maria Schuld, Ilya Sinayskiy, and Francesco Petruccione.
\newblock Prediction by linear regression on a quantum computer.
\newblock {\em Physical Review A}, 94(2):022342, 2016.

\bibitem{Rebentrost14}
Patrick Rebentrost, Masoud Mohseni, and Seth Lloyd.
\newblock Quantum support vector machine for big data classification.
\newblock {\em Phys. Rev. Lett.}, 113:130503, Sep 2014.

\bibitem{Harrow09}
Aram~W Harrow, Avinatan Hassidim, and Seth Lloyd.
\newblock Quantum algorithm for linear systems of equations.
\newblock {\em Physical Review Letters}, 103(15):150502, 2009.

\bibitem{peruzzo14}
Alberto Peruzzo, Jarrod McClean, Peter Shadbolt, Man-Hong Yung, Xiao-Qi Zhou,
  Peter~J Love, Al{\'a}n Aspuru-Guzik, and Jeremy~L O’brien.
\newblock A variational eigenvalue solver on a photonic quantum processor.
\newblock {\em Nature communications}, 5, 2014.

\bibitem{farhi14}
Edward Farhi, Jeffrey Goldstone, and Sam Gutmann.
\newblock A quantum approximate optimization algorithm.
\newblock {\em arXiv preprint arXiv:1411.4028}, 2014.

\bibitem{Lloyd14}
Seth Lloyd, Masoud Mohseni, and Patrick Rebentrost.
\newblock Quantum principal component analysis.
\newblock {\em Nature Physics}, 10:631--633, 2014.

\bibitem{Berry2012}
Dominic~W. Berry and Andrew~M. Childs.
\newblock Black-box hamiltonian simulation and unitary implementation.
\newblock {\em Quantum Info. Comput.}, 12(1-2):29--62, January 2012.

\bibitem{jordan05}
Stephen~P Jordan.
\newblock Fast quantum algorithm for numerical gradient estimation.
\newblock {\em Physical Review Letters}, 95(5):050501, 2005.

\bibitem{Kerenidis2017}
I.~Kerenidis and A.~Prakash.
\newblock Quantum gradient descent for linear systems and least squares.
\newblock {\em arXiv:1704.04992}, 2017.

\bibitem{Gilyen2017}
Andr\'as Gily\'en, Srinivasan Arunachalam, and Nathan Wiebe.
\newblock Optimizing quantum optimization algorithms via faster quantum
  gradient computation.
\newblock {\em arXiv:1711.00465}, 2017.

\bibitem{he10}
Simai He, Zhening Li, and Shuzhong Zhang.
\newblock Approximation algorithms for homogeneous polynomial optimization with
  quadratic constraints.
\newblock {\em Mathematical Programming}, 125(2):353--383, 2010.

\bibitem{Goldstein64}
Alan~A Goldstein.
\newblock Convex programming in hilbert space.
\newblock {\em Bulletin of the American Mathematical Society}, 70(5):709--710,
  1964.

\bibitem{Levitin66}
Evgeny~S Levitin and Boris~T Polyak.
\newblock Constrained minimization methods.
\newblock {\em USSR Computational mathematics and mathematical physics},
  6(5):1--50, 1966.

\bibitem{Giovannetti2008}
V.~Giovannetti, S.~Lloyd, and L.~Maccone.
\newblock Quantum random access memory.
\newblock {\em Phys. Rev. Lett.}, 100:160501, 2008.

\bibitem{Giovannetti2008_2}
V.~Giovannetti, S.~Lloyd, and L.~Maccone.
\newblock Architectures for a quantum random access memory.
\newblock {\em Phys. Rev. A}, 78:052310, 2008.

\bibitem{Martini2009}
F.~De Martini, V.~Giovannetti, S.~Lloyd, L.~Maccone, E.~Nagali, L.~Sansoni, and
  F.~Sciarrino.
\newblock Experimental quantum private queries with linear optics.
\newblock {\em Phys. Rev. A}, 80:010302, 2009.

\bibitem{Grover2002}
L.~Grover and T.~Rudolph.
\newblock Creating superpositions that correspond to efficiently integrable
  probability distributions.
\newblock {\em arXiv preprint quant-ph/0208112}, 2002.

\bibitem{Soklakov06}
Andrei~N Soklakov and R{\"u}diger Schack.
\newblock Efficient state preparation for a register of quantum bits.
\newblock {\em Physical Review A}, 73(1):012307, 2006.

\bibitem{Hoeffding1994}
W.~Hoeffding and H.~Robbins.
\newblock The central limit theorem for dependent random variables.
\newblock In {\em Fisher N.I., Sen P.K. (eds) The Collected Works of Wassily
  Hoeffding. Springer Series in Statistics (Perspectives in Statistics)}, New
  York, NY, 1994. Springer.

\bibitem{Hinton06}
Geoffrey~E Hinton, Simon Osindero, and Yee-Whye Teh.
\newblock A fast learning algorithm for deep belief nets.
\newblock {\em Neural computation}, 18(7):1527--1554, 2006.

\bibitem{Bengio09}
Yoshua Bengio.
\newblock Learning deep architectures for ai.
\newblock {\em Foundations and trends{\textregistered} in Machine Learning},
  2(1):1--127, 2009.

\bibitem{Kimmel2017}
Shelby Kimmel, Cedric Yen-Yu Lin, Guang~Hao Low, Maris Ozols, and Theodore~J
  Yoder.
\newblock Hamiltonian simulation with optimal sample complexity.
\newblock {\em npj Quantum Information}, 3(1):13, 2017.

\bibitem{kimmel2017hamiltonian}
Shelby Kimmel, Cedric Yen-Yu Lin, Guang~Hao Low, Maris Ozols, and Theodore~J
  Yoder.
\newblock Hamiltonian simulation with optimal sample complexity.
\newblock {\em npj Quantum Information}, 3(1):13, 2017.

\bibitem{Childs2017}
Andrew~M. Childs, Dmitri Maslov, Yunseong Nam, Neil~J. Ross, and Yuan Su.
\newblock Toward the first quantum simulation with quantum speedup.
\newblock {\em arXiv:1711.10980}, 2017.

\bibitem{Bertsekas1982}
Dimitri~P. Bertsekas.
\newblock Projected newton methods for optimization problems with simple
  constraints.
\newblock {\em SIAM J. Control and Optimization}, 20(2):221, 1982.

\bibitem{Lee2012}
Jason~D. Lee, Yuekai Sun, and Michael~A. Saunders.
\newblock Proximal newton-type methods for convex optimization.
\newblock In {\em Proceedings of the 25th International Conference on Neural
  Information Processing Systems - Volume 1}, NIPS'12, pages 827--835, USA,
  2012. Curran Associates Inc.

\bibitem{Childs2015}
A~M. Childs, R.~Kothari, and R.~D. Somma.
\newblock Quantum linear systems algorithm with exponentially improved
  dependence on precision.
\newblock {\em arXiv:1511.02306}, 2015.

\bibitem{Low2016qubitization}
Guang~Hao Low and Isaac~L. Chuang.
\newblock Hamiltonian simulation by qubitization.
\newblock {\em arXiv:1610.06546}, 2016.

\bibitem{Low2017spectral}
Guang~Hao Low and Isaac~L. Chuang.
\newblock Hamiltonian simulation by uniform spectral amplification.
\newblock {\em arXiv:1707.05391}, 2017.

\bibitem{Brassard2002}
Gilles Brassard, Peter Hoyer, Michele Mosca, and Alain Tapp.
\newblock Quantum amplitude amplification and estimation.
\newblock {\em Contemporary Mathematics}, 305:53--74, 2002.

\bibitem{Ambainis2010}
A.~Ambainis.
\newblock Variable time amplitude amplification and a faster quantum algorithm
  for solving systems of linear equations.
\newblock {\em arXiv:1010.4458}, 2010.

\bibitem{Brandao2017}
Fernando G. S. L.~Brand\ ao, Amir Kalev, Tongyang Li, Cedric Yen-Yu Lin,
  Krysta~M. Svore, and Xiaodi Wu.
\newblock Exponential quantum speed-ups for semidefinite programming with
  applications to quantum learning.
\newblock {\em arXiv:1710.02581}, 2017.

\bibitem{Apeldoorn2017}
Joran van Apeldoorn, Andr\'as Gily\'en, Sander Gribling, and Ronald de~Wolf.
\newblock Quantum sdp-solvers: Better upper and lower bounds.
\newblock {\em arXiv:1705.01843}, 2017.

\bibitem{bookatz2012qma}
Adam~D Bookatz.
\newblock Qma-complete problems.
\newblock {\em arXiv preprint arXiv:1212.6312}, 2012.

\bibitem{dauphin2014identifying}
Yann~N Dauphin, Razvan Pascanu, Caglar Gulcehre, Kyunghyun Cho, Surya Ganguli,
  and Yoshua Bengio.
\newblock Identifying and attacking the saddle point problem in
  high-dimensional non-convex optimization.
\newblock In {\em Advances in neural information processing systems}, pages
  2933--2941, 2014.

\bibitem{cubitt2017universal}
Toby Cubitt, Ashley Montanaro, and Stephen Piddock.
\newblock Universal quantum hamiltonians.
\newblock {\em arXiv preprint arXiv:1701.05182}, 2017.

\end{thebibliography}
\end{document}